%% file: main.tex
\documentclass[11pt]{article}
\usepackage{etex}

\usepackage[top=1in, bottom=1in, left=1in, right=1in]{geometry}

\usepackage{amsmath}
\usepackage{amssymb}
\usepackage{amsthm}
\usepackage{array}
\usepackage{bbm}
\usepackage{bm}
\usepackage{cancel}
\usepackage{cmap}
\usepackage{csquotes}
\usepackage{enumerate}
\usepackage{enumitem}
\usepackage{esint}
\usepackage{fancyhdr}
\usepackage{listings}
\usepackage{longtable}
\usepackage{makeidx}
\usepackage{mathdots}
\usepackage{mathtools}
\usepackage{mathrsfs}
\usepackage{stackrel}
\usepackage{stmaryrd}
\usepackage{tabularx}
\usepackage{tikz}
\usepackage{ctable}
\usepackage{titlesec}
\usepackage{titletoc}
\usepackage{url}
\usepackage{verbatim}
\usepackage{wasysym}
\usepackage{wrapfig}
\usepackage{yhmath}
\usepackage[all,cmtip]{xy}

\usepackage{hyperref}
\usepackage[capitalise]{cleveref}

\usepackage{algorithm}
\usepackage{algpseudocode}
\algnewcommand\algorithmicinput{\textbf{Input: }}
\algnewcommand\INPUT{\State\algorithmicinput}
\algnewcommand\algorithmicinitialize{\textbf{Initialize: }}
\algnewcommand\INIT{\State\algorithmicinitialize}
\algnewcommand\algorithmicrun{\textbf{Run: }}
\algnewcommand\RUN{\State\algorithmicrun}
\algnewcommand\algorithmicupdate{\textbf{Update: }}
\algnewcommand\UPDATE{\State\algorithmicupdate}
\algnewcommand\algorithmicset{\textbf{Set: }}
\algnewcommand\SET{\State\algorithmicset}
\algnewcommand\algorithmicquery{\textbf{Query: }}
\algnewcommand\QUERY{\State\algorithmicquery}
\algnewcommand\algorithmicoutput{\textbf{Output: }}
\algnewcommand\OUTPUT{\State\algorithmicoutput}

\usepackage[backend=biber,style=alphabetic,maxbibnames=99]{biblatex}

\usetikzlibrary{calc,trees,positioning,arrows,chains,shapes.geometric,%
    decorations.pathreplacing,decorations.pathmorphing,shapes,%
    matrix,shapes.symbols,shadows,fadings}

\usepackage[refpage]{nomencl}

\makenomenclature

\makeindex

\input{theorems_with_boxes.tex}
\input{macros.tex}

\input{formatting.tex}

\usepackage{palatino}

\renewcommand{\KL}[0]{\mathcal{D}_{\textup{KL}}}
\newcommand{\kabl}[1]{\ka_{\textup{BL},#1}}

\newcommand{\Dn}[2]{D_{#1\to #2}}
\newcommand{\Up}[2]{U_{#1\to #2}}

\newcommand{\Pudl}[1]{P_{#1}^{\triangle}}

\newcommand{\Pdul}[1]{P_{#1}^{\triangledown}}
\newcommand{\Pdun}[1]{P_{#1}^{\vee}}

\title{Parallelising Glauber dynamics}
\author{Holden Lee\footnote{Johns Hopkins University, \texttt{hlee283@jhu.edu}}}
\date{\today}

\renewcommand{\blu}[1]{#1}

\begin{document}
\maketitle

\input{body}

\printbibliography
\appendix
\input{comparison}

\end{document}

%% file: theorems_with_boxes.tex
\newtheoremstyle{norm}
{12pt}
{12pt}
{}
{}
{\bf}
{:}
{.5em}
{}

\newtheorem{thm}{Theorem}[section]
\newtheorem*{thm*}{Theorem}

\newtheorem*{clm*}{Claim}

\newtheorem*{conj*}{Conjecture}
\newtheorem{cor}[thm]{Corollary}
\newtheorem{lem}[thm]{Lemma}
\newtheorem*{lem*}{Lemma}

\theoremstyle{norm}
\newtheorem{prb}[thm]{Problem}
\newtheorem*{prb*}{Problem}

\newtheorem*{ax*}{Axiom}
\newtheorem{df}[thm]{Definition}
\newtheorem*{df*}{Definition}

\newtheorem*{ex*}{Example}

\newtheorem{expl}[thm]{Exploration}

\newtheorem*{pos*}{Postulate}
\newtheorem{pr}[thm]{Proposition}
\newtheorem*{pr*}{Proposition}

\newtheorem*{qu*}{Question}
\newtheorem{rem}[thm]{Remark}
\newtheorem*{rem*}{Remark}

\usepackage[english]{babel} 
\usepackage{blindtext} 
\tikzstyle{prbox} = [draw=black, fill=blue!20, very thick,
    rectangle, inner sep=10pt, inner ysep=10pt]
\tikzstyle{thbox} = [draw=black,double, fill=blue!10, very thick,
    rectangle, inner sep=10pt, inner ysep=10pt]
\tikzstyle{cpbox} = [drop shadow={
    shadow scale=1}, draw=black, fill=green!10, very thick,
    rectangle, inner sep=10pt, inner ysep=10pt]
\tikzstyle{wrbox} = [drop shadow={
    shadow scale=1}, draw=black, fill=yellow!10, very thick,
    rectangle, inner sep=10pt, inner ysep=10pt]
\tikzstyle{hnbox} = [draw=black, fill=white, very thick,
    rectangle, inner sep=10pt, inner ysep=10pt]

%% file: macros.tex
\newcommand{\sE}[0]{\mathscr{E}}
\newcommand{\E}[0]{\mathbb{E}}

\newcommand{\Pj}[0]{\mathbb{P}}

\newcommand{\R}[0]{\mathbb{R}}

\newcommand{\one}[0]{\mathbbm{1}}

\newcommand{\be}[0]{\beta}
\newcommand{\ga}[0]{\gamma}
\newcommand{\Ga}[0]{\Gamma}
\newcommand{\de}[0]{\delta}

\newcommand{\ep}[0]{\varepsilon}
\newcommand{\eph}[0]{\frac{\varepsilon}{2}}

\newcommand{\ka}[0]{\kappa}
\newcommand{\la}[0]{\lambda}

\newcommand{\rh}[0]{\rho}

\newcommand{\Te}[0]{\Theta}

\newcommand{\Om}[0]{\Omega}
\newcommand{\si}[0]{\sigma}

\newcommand{\nin}[0]{\not\in}
\newcommand{\opl}[0]{\oplus}

\newcommand{\ot}[0]{\otimes}

\newcommand{\sub}[0]{\subset}

\newcommand{\subeq}[0]{\subseteq}
\newcommand{\supeq}[0]{\supseteq}

\newcommand{\bs}[0]{\backslash}
\newcommand{\iy}[0]{\infty}

\newcommand{\rc}[1]{\frac{1}{#1}}
\newcommand{\prc}[1]{\pa{\rc{#1}}}

\newcommand{\fc}[2]{\frac{#1}{#2}}

\newcommand{\pf}[2]{\pa{\frac{#1}{#2}}}

\newcommand{\dd}[2]{\frac{d #1}{d #2}}

\newcommand{\nb}[0]{\nabla}

\newcommand{\lra}[0]{\leftrightarrow}

\newcommand{\ab}[1]{\left| {#1} \right|}
\newcommand{\an}[1]{\left\langle {#1}\right\rangle}
\newcommand{\ba}[1]{\left[ {#1} \right]}
\newcommand{\bc}[1]{\left\{ {#1} \right\}}

\newcommand{\ce}[1]{\left\lceil {#1}\right\rceil}
\newcommand{\fl}[1]{\left\lfloor {#1}\right\rfloor}

\newcommand{\pa}[1]{\left( {#1} \right)}

\newcommand{\ve}[1]{\left\Vert {#1}\right\Vert}

\newcommand{\set}[2]{\left\{{#1}:{#2}\right\}}

\newcommand{\ol}[1]{\overline{#1}}

\newcommand{\ub}[2]{\underbrace{#1}_{#2}}

\newcommand{\wt}[1]{\widetilde{#1}}
\newcommand{\wh}[1]{\widehat{#1}}

\newcommand{\Ent}{\operatorname{Ent}}

\renewcommand{\hom}{\textrm{hom}}

\newcommand{\KL}[0]{\operatorname{KL}}

\newcommand{\poly}{\operatorname{poly}}

\newcommand{\sgn}{\operatorname{sign}}

\newcommand{\Tr}[0]{\operatorname{Tr}}

\newcommand{\TV}[0]{\mathcal D_{\mathrm{TV}}}

\newcommand{\Var}[0]{\operatorname{Var}}

\newcommand{\vc}[0]{\operatorname{vec}}

\providecommand{\cal}[1]{\mathcal{#1}}
\renewcommand{\cal}[1]{\mathcal{#1}}

\newcommand{\pull}[9]{
#1\ar@/_/[ddr]_{#2} \ar@{.>}[rd]^{#3} \ar@/^/[rrd]^{#4} & &\\
& #5\ar[r]^{#6}\ar[d]^{#8} &#7\ar[d]^{#9} \\}

\newcommand{\cmp}[9]{
\xymatrix{
#1 \ar[r]^{#4}{#5} \ar@/_2pc/[rr]^{#8}_{#9} & #2 \ar[r]^{#6}_{#7} & #3
}
}

\newcommand{\ha}[1]{\ar@{^(->}[#1]}
\newcommand{\ls}[1]{\ar@{-}[#1]}
\newcommand{\sj}[1]{\ar@{->>}[#1]}
\newcommand{\aq}[1]{\ar@{=}[#1]}
\newcommand{\acir}[1]{\ar@{}[#1]|-{\textstyle{\circlearrowright}}}
\newcommand{\acil}[1]{\ar@{}[#1]|-{\textstyle{\circlearrowleft}}}
\newcommand{\ard}[1]{\ar@{.>}[#1]}
\newcommand{\mt}[1]{\ar@{|->}[#1]}
\newcommand{\inm}[1]{\ar@{}[#1]|-{\in}}
\newcommand{\inr}{\ar@{}[d]|-{\rotatebox[origin=c]{-90}{$\in$}}}
\newcommand{\inl}{\ar@{}[u]|-{\rotatebox[origin=c]{90}{$\in$}}}

\newcommand{\sumr}[2]{\sum_{\scriptsize \begin{array}{c}{#1}\\{#2}\end{array}}}

\newcommand{\sumo}[2]{\sum_{#1=1}^{#2}}

\newcommand{\prodo}[2]{\prod_{#1=1}^{#2}}

\newcommand{\beq}[1]{\begin{equation}\llabel{#1}}
\newcommand{\eeq}[0]{\end{equation}}
\newcommand{\bal}[0]{\begin{align*}}
\newcommand{\eal}[0]{\end{align*}}
\newcommand{\ban}[0]{\begin{align}}
\newcommand{\ean}[0]{\end{align}}

\newcommand{\blu}[1]{{\color{blue}#1}}

\newcommand{\fixme}[1]{{\color{red}#1}}
\newcommand{\llabel}[1]{\label{#1}\text{\fixme{\tiny#1}}}

\newcommand{\arxiv}[1]{\url{http://www.arxiv.org/abs/#1}}

\newcommand{\vocab}[1]{\textbf{#1}} 

\allowdisplaybreaks[2]

\DeclareFontFamily{U}{wncy}{}
    \DeclareFontShape{U}{wncy}{m}{n}{<->wncyr10}{}
    \DeclareSymbolFont{mcy}{U}{wncy}{m}{n}
    \DeclareMathSymbol{\Sh}{\mathord}{mcy}{"58}

\newcommand{\Bn}[0]{\{\pm1\}^n}
\newcommand{\seti}[3]{{#1}_{#2\leftarrow#3}}

\newcommand{\D}[0]{\operatorname{D}}

\newcommand{\CP}[0]{C_{\textrm{P}}}

\newcommand{\chis}[0]{\cal D_{\chi^2}}
\newcommand{\Df}[0]{\cal D_f}

\renewcommand{\sE}[0]{\mathcal E}

\newcommand{\estep}[0]{\varepsilon_{\mathrm{step}}}

%% file: body.tex
\begin{abstract}
    For distributions over discrete product spaces $\prod_{i=1}^n \Omega_i'$, Glauber dynamics is a Markov chain that at each step, resamples a random coordinate conditioned on the other coordinates. We show that $k$-Glauber dynamics, which resamples a random subset of $k$ coordinates, mixes $k$ times faster in $\chi^2$-divergence, and assuming approximate tensorization of entropy, mixes $k$ times faster in KL-divergence. 
    We apply this to obtain parallel algorithms in two settings: (1) For the Ising model $\mu_{J,h}(x)\propto \exp(\frac1 2\left\langle x,Jx \right\rangle + \langle h,x\rangle)$ with
    $\|J\|<1-c$ (the regime where fast mixing is known), we show that we can implement each step of $\widetilde \Theta(n/\|J\|_F)$-Glauber dynamics efficiently with a parallel algorithm, resulting in a parallel algorithm with running time $\widetilde O(\|J\|_F) = \widetilde O(\sqrt n)$. (2) For the mixed $p$-spin model at high enough temperature, we show that with high probability we can implement each step of $\wt \Theta(\sqrt n)$-Glauber dynamics efficiently and obtain running time $\wt O(\sqrt n)$.
\end{abstract}

\section{Introduction}

A key problem in computer science and statistics is to sample from a probability distribution given its probability mass function up to a constant of proportionality. 
The problem has been studied both over discrete spaces (such as $\Om^n$ for a finite set $\Om$) and continuous spaces (such as $\R^n$); the goal is to give efficient algorithms for general classes of distributions, and in particular, to obtain optimal scaling in the dimension $n$. In this work we focus on minimizing the parallel running time, assuming a polynomial number of processors. 
In $\R^n$, it is natural to change multiple coordinates at a time using gradient-based algorithms such as Langevin dynamics and Hamiltonian Monte Carlo; many results have given algorithms that require a sublinear number of steps for log-concave distributions in various settings.  

However, on discrete product spaces $\Om^n$, the canonical algorithm, Glauber dynamics, involves resampling coordinates one at a time, and hence requires at least $n$ steps in general. 
A natural attempt to speed up Glauber dynamics with parallel computation is to resample $k$ coordinates at a time. We establish that under general conditions, this simple idea does indeed speed up Glauber dynamics by a factor of approximately $k$.

To obtain a parallel algorithm, the task remains to give a fast parallel method of resampling $k$ coordinates. We show that this can be done in the case of the Ising model $\mu_{J,h}$ over $\Bn$ when the interaction matrix $J$ is bounded away from 1 in operator norm, $\ve{J}<1-c$, and in the case of the mixed $p$-spin model at high enough temperature, both of which are known to enjoy rapid mixing of standard Glauber dynamics.

The Ising model is a classical model from statistical physics which has probability mass function on $\Bn$ given by
\begin{align*}
\mu_{J,h}(x) &= \rc{Z_{J,h}}\exp\pa{\rc 2 \an{x,Jx} + \an{h,x}},
&\text{where }Z_{J,h} = \sum_{x\in \Bn} \exp\pa{\rc 2 \an{x,Jx} + \an{h,x}}.
\end{align*}
The regime $\ve{J}<1$ is exactly where (based on information of the operator norm alone) Glauber dynamics is known to have fast mixing~\cite{eldan2022spectral,anari2021entropic}. 
To sample $k$ coordinates, we use approximate rejection sampling with a product distribution and a further recursion for certain ``bad" sets. By taking $k = \wt \Te(n/\ve{J}_F)$, we obtain an algorithm with parallel running time $\wt O(\ve{J}_F) = \wt O(\sqrt n)$. 

The \vocab{mixed $p$-spin model} with coefficients $\be_2,\be_3,\ldots$ and external field $h\in \R^n$ is the random measure on $\{\pm 1\}^n$ given by\footnote{The factor $\sqrt{p!}$ arises as we index only over increasing sequences.}
\begin{multline}
\label{e:p-spin}
\mu_{\be, g,h}(x) \propto 
\exp(H_{\be, g, h}(x)), 
\quad \text{where }
H_{\be, g, h}(x) = \sum_{p=2}^{\iy}\fc{\be_p\sqrt{p!}}{n^{\fc{p-1}2}} \sum_{1\le i_1<\cdots < i_p\le n} g_{i_1,\cdots ,i_p}x_{i_1}\cdots x_{i_p} + \sumo in h_ix_i 
\end{multline}
and $g_{i_1,\ldots, i_p}\sim N(0,1)$. By taking $k=\Te(\sqrt n)$, we obtain an algorithm with parallel running time $\wt O(\sqrt n)$.

\subsection{Main results}
Let $\mu$ be a distribution on $\prodo in\Om_i'$. 
We define $k$-Glauber dynamics  as the Markov chain which given a sample $x\in \prodo in\Om_i'$, chooses a subset $S\subeq [n]$ uniformly at random among subsets of size $k$, and resamples the coordinates in $S$ conditional on coordinates in $S^c$, according to the distribution of $\mu$. 
Let $P_{\mu, k}$ denote its Markov kernel. 

We show that under general conditions, $k$-Glauber dynamics mixes $k$ times faster in both $\chi^2$ and KL-divergence. We say that a Markov kernel $P$ with stationary distribution $\mu$ satisfies $\rh$-contraction in $\chi^2$-divergence if $\chis(\nu P\|\mu)\le \rh \chis(\nu\|\mu)$ and similarly for $\KL$; this can be iterated to give a mixing time bound. See \Cref{s:fi} for background on functional inequalities (Poincar\'e inequality and approximate tensorization of entropy).
\begin{thm}[$k$-Glauber mixes $k$ times faster]\label{t:main-fast}
Let $\mu$ be a distribution on $\Om = \prodo in \Om_i'$, and let $1\le k\le n$. 
Below, let $C\ge 1$.
\begin{enumerate}
    \item If $\mu$ satisfies a Poincar\'e inequality with constant $Cn$, then $P_{\mu,k}$ satisfies a Poincar\'e inequality with constant $O\pf{Cn}{k}$, and satisfies $(1-\Om\pf{k}{Cn})$-contraction in $\chi^2$-divergence.
    \item If $\mu$ satisfies $C$-approximate tensorization of entropy (so that $P_{\mu}$ satisfies $(1-\Om\pf{1}{Cn})$-contraction in KL-divergence), then $P_{\mu,k}$ satisfies 
    $(1-\Om\pf{k}{Cn})$-contraction in KL-divergence.
\end{enumerate}
\end{thm}
Here, the $O(\cdot)$ and $\Om(\cdot)$ hide only universal constants.
The Poincar\'e inequality is equivalent to contraction in $\chi^2$-divergence, so part (1) gives a $\Om(k)$-factor speedup to mixing in $\chi^2$. 
The analogue of the Poincar\'e inequality for KL is a modified log-Sobolev inequality. Although we need the slightly stronger notion of approximate tensorization of entropy to prove a speedup to mixing in KL, we note that many works that establish a modified log-Sobolev inequality do so using tensorization of entropy~\cite{blanca2022mixing,anari2021entropic}. See Section~\ref{s:prelim} for relevant background on mixing for Markov chains.

We prove Theorem~\ref{t:main-fast} as \cref{c:k-glauber} of the more general \cref{t:du}. We view $k$-Glauber dynamics as randomly erasing $k$ coordinates one by one, and then adding them back one by one according to the right conditional distributions. This realizes $k$-Glauber dynamics as a composition of down and up operators $\Dn{n}{n-1}\cdots \Dn{n-k+1}{n-k}\Up{n-k}{n-k+1}\cdots \Up{n-1}n$. The assumptions give contraction of $\Dn{n}{n-1}$, and our general theorem shows that the contraction of each $\Dn{j}{j-1}$ is at least as good as $\Dn{n}{n-1}$ (except for an additive factor). To do this, we realize $\Dn{j}{j-1}$ as $\Dn{n}{n-1}$ tensorized with erasure ``noise" and projected, and bound how the factor of contraction changes under these operations. We make an analogy to bounding the Poincar\'e and log-Sobolev constants of a distribution $\mu$ on $\R^n$ convolved with Gaussian noise, and the proximal sampler based on iteratively adding and removing Gaussian noise (more specifically, sampling from the posterior distribution given a noisy Gaussian observation of the sample from $\mu$).

Algorithmically, the challenge with implementing $k$-Glauber dynamics is that naive enumeration for the transition kernel takes $2^k$ time, and hence we must find a way to use the structure of the distribution to implement each step more efficiently. 
We show that in the case of the Ising model, we can efficiently simulate $k$-Glauber dynamics for $k=\wt O\pa{{n}/{\ve{J^{\invdiameter}}_F}}$, to obtain a parallel algorithm running in time $\wt O\pa{\ve{J^{\invdiameter}}_F}$, where $J^{\invdiameter}$ denotes $J$ with diagonal entries set to 0\footnote{While changing the diagonal entries of $J$ does not change the Ising model, we need to allow $J$ to have nonzero diagonal entries in order to be positive semi-definite.}. Under the assumption that $\ve{J}<1$, this is always at most $\wt O(\sqrt n)$.
\begin{thm}\label{t:main}
Let $c>0$. 
With appropriate choice of constants depending only on $c$, if $J$ is symmetric positive semi-definite with $\ve{J}\le 1-c$, then $\mathsf{ParallelIsingSampler}$  (Algorithm~\ref{a:pising}) with appropriate constants
outputs a sample $\ep$-close in TV distance from the Ising model $\mu_{J,h}$ and, 
with probability at least $1-\ep$, 
runs in time
$O\pa{\max\bc{\ve{J^{\invdiameter}}_F,1}\poly\log\pf n\ep}$ on a parallel machine with $\poly(n)$ processors. 
\end{thm}
We note that our algorithm is a high-accuracy sampler: the only dependence on $\ep$ is a poly-logarithmic dependence in the running time. Notably, the number of processors does not depend on $\ep$.
We rely on the result~\cite{anari2021entropic} that gives optimal ($O(n\ln n)$) mixing times for the Ising model for $\ve{J}<1$ based on the theory of entropic independence. 

The first attempt to implement $k$-Glauber dynamics is to approximate the conditional distribution of $k$ coordinates using a carefully chosen product distribution and use rejection sampling. Using concentration results (the Hanson-Wright inequality), if $\ve{J_{S\times S}^{\invdiameter}}_F$ is small for the randomly chosen set $S$, then this succeeds with high probability. The complication is that $\ve{J_{S\times S}^{\invdiameter}}_F$ can sometimes be large. If this is the case, then we recurse on $J_{S\times S}$. By controlling the expected size of $\ve{J_{S\times S}^{\invdiameter}}_F$, we show that the recursive calls form a subcritical branching process and with high probability, add at most a polylogarithmic overhead to the running time. 

\newcommand{\frobslack}[0]{\ln\pf 2{\ep_{\mathrm{step}}} + 1}
\newcommand{\pfrobslack}[0]{\pa{\frobslack}}
\newcommand{\qars}[0]{\mathsf{QuadraticApproxRejectionSampler}}

\begin{algorithm}[h!]
\caption{Parallel Ising Sampler ($\mathsf{ParallelIsingSampler}$)}
\begin{algorithmic}[1]
\INPUT Interaction matrix $J\in \R^{n\times n}$, subset $R$ of size $m$, external field $h\in \R^R$, error parameter $\ep\in (0,\rc2)$. 
\State Let $\ep_{\mathrm{step}}= \fc{\ep}{2n^{C_4}}$.
\If{$\ve{J_{R\times R}^{\invdiameter}}_F\le \fc{c_3}{\frobslack}$ ($J^{\invdiameter}$ denotes $J$ with diagonal entries set to 0)}{}
    \State $y\mapsfrom \mathsf{QuadraticApproxRejectionSampler}\pa{H(x) = \rc 2 \an{x,J_{R\times R}x} + \an{h, x}, \fc{c_3}{\frobslack}, \ep_{\textrm{step}}}$. (See Algorithm~\ref{a:qars}.)
\Else{}
    \State Let $s =\ce{\fc{c_1m}{\pfrobslack\ln \pf{n}{\ep} \ve{J_{R\times R}^{\invdiameter}}_F}}$.
    \State Let $T =\fl{ C_2 \ln \pf{n}{\ep} \fc{m}{s}}$. \label{st:T}
    \State Draw $y$ from the product distribution $\nu_0(x)\propto e^{\an{h,x}}$. 
    \For{$t$ from 1 to $T$}
        \State Choose $S\subeq R$ a random subset of size $s$.
        \State{$z\mapsfrom \mathsf{ParallelIsingSampler}(J, S,  J_{S\times R\bs S}y_{R\bs S}+ h_S, \ep)$}
    \State Set $y_S = z$.
\EndFor
\EndIf
\OUTPUT $y$ (Approximate sample from $\mu_{J_{R\times R},h}$).
\end{algorithmic}
\label{a:pising}
\end{algorithm}

\begin{thm}\label{t:main-p-spin}
Consider the mixed $p$-spin model~\eqref{e:p-spin}. There exists an absolute constant $\de>0$ such that if $\sum_{p\ge 2} \sqrt{p^3\ln p}\cdot \be_p<\de$ and $D(\be)=\sum_{p\ge 2} \sqrt{2^pp^3\ln p}\cdot  \be_p<\iy$, then with probability $1-\exp(-\Om(n))$ over $g$, given query access to $H_{\be, g,h}$, Algorithm~\ref{a:ppspin} outputs a sample $\ep$-close in TV distance from $\mu_{\be, g, h}$ and, with probability at least $1-\ep$, runs in time $O_{D(\be)} \pa{\sqrt n \poly\log\pf{n}{\ep}}$ on a parallel machine with $\poly\pf{n}{\ep}$ processors.
\end{thm}
Note that a recursion is not necessary in Algorithm~\ref{a:ppspin}. Intuitively, the mean-field nature of the $p$-spin model ensures that with high probability all marginal distributions of $O(\sqrt n)$ coordinates are well-approximated by a product distribution. Though we do not investigate this further, a recursive algorithm could potentially eliminate the $\poly(1/\ep)$ dependence on the number of processors as in Theorem~\ref{t:main}.

\begin{algorithm}[h!]
\caption{Parallel $p$-spin Sampler 
}
\begin{algorithmic}[1]
\INPUT Access to $H=H_{\be, g, h}$ (Hamiltonian) \eqref{e:p-spin}, 
error parameter $\ep\in (0,\rc2)$. 
\State Let $T=\fl{ \fc{C_2}{c_1} \ln \pf{n}{\ep} \sqrt n}$.
\State Let $\ep_{\mathrm{step}}= \fc{\ep}{2T}$.
\State Draw $y$ from the product distribution $\nu_0(x)\propto e^{\an{h,x}}$. 
\For{$t$ from 1 to $T$}
    \State Choose $S\subeq R$ a random subset of size $\ce{c_1\sqrt n}$.
    \State $z\mapsfrom \mathsf{QuadraticApproxRejectionSampler}(H_{x_{S^c}},c_3/2,\ep_{\textrm{step}})$.
    \label{st:call-qars-pspin}
    \State Set $y_S = z$.
\EndFor
\OUTPUT $y$ (Approximate sample from $\mu_{\be, g, h}$).
\end{algorithmic}
\label{a:ppspin}
\end{algorithm}

We view our result on the Ising model and the $p$-spin model as proofs of concept for parallelisation using $k$-Glauber dynamics, and hope it serves as a useful framework for constructing parallel algorithms for other families of discrete distributions. As discussed in the next section, using a different parallel algorithm, the work \cite{liu2022simple} obtains \Cref{t:main} but not \Cref{t:main-p-spin}.

\subsection{Related work}

\blu{We note that our Theorem~\ref{t:main-fast} can be viewed as a complement of ``local-to-global" results for mixing of the down-up walks~\cite{oppenheim2018local,alev2020improved,chen2021optimal}, and is not implied by those results. Those results aim to establish mixing of Glauber dynamics (or the down-up walk) from mixing of simpler chains, while we start by assuming mixing of Glauber dynamics. In particular, \cite{chen2021optimal} apply the reverse strategy: for the spin systems on graphs they consider, they show that mixing of $\theta n$-Glauber dynamics (for appropriate $\theta$) implies mixing of Glauber dynamics.

When contraction of Glauber dynamics is derived directly from either spectral or entropic independence using local-to-global arguments, then the same arguments can be used to establish mixing of the $k$-Glauber (e.g., using $k$-uniform block factorization of entropy~\cite{chen2021optimal}, the analogue of approximate tensorization of entropy). However, this does not apply for distributions for which mixing is established through other methods. 
The recent work~\cite{anari2023universality} shows that a Poincar\'e inequality implies spectral independence, but the bound obtained for $k$-Glauber through spectral independence is lossy (resulting in a power of $n$). 
Our work can be seen as giving a general conceptual reason why mixing for Glauber must imply mixing for $k$-Glauber.
}


\subsubsection{Continuous sampling}

For log-concave distributions on $\R^n$, a long line of works on the underdamped Langevin algorithm and Metropolis-adjusted Langevin algorithm have led to high-accuracy sampling using $\wt O(n^{1/2})$ steps \cite{altschuler2023faster}. The randomized midpoint method for underdamped Langevin dynamics allows sampling in the weaker Wasserstein metric in $\wt O(n^{1/3})$ steps \cite{shen2019randomized}, and can furthermore be fully parallelised to obtain $\ep$ error with $\poly\pf{n}{\ep}$ processors. 
These dependencies are assuming the condition number is $O(1)$. 

We note that the Ising model for $\ve{J}<1$ can be decomposed as a log-concave mixture of product distributions \cite{hubbard1959calculation,bauerschmidt2019very,koehler2022sampling}, so these algorithms give an alternative approach to parallel algorithms for the Ising model. However, this decomposition is highly specific to the Ising model. Moreover, the Wasserstein guarantee is incompatible with a TV guarantee, and the complexity of our approach scales with $\ve{J}_F$.




\subsubsection{Parallel algorithms for discrete sampling}

Recent work \cite{anari2020sampling,anari2022improved,anari2023parallel} has investigated the question of obtaining fast parallel algorithms for approximate sampling in settings where fast parallel algorithms for approximate \emph{counting} (or computing a partition function) exist. In particular, for distributions satisfying transport stability and where the \emph{log-Laplace} transform can be efficiently calculated (e.g., using the efficient algorithm for computing partition functions), \cite{anari2023parallel} gives a $\poly\log(n/\ep)$-time algorithm with $\poly(n/\ep)$ many processors (i.e., a $\mathsf{RNC}$ algorithm). 
This includes problems such as determinantal point processes and Eulerian tours. Notably they use the continuous algorithm (randomized midpoint method, discussed above) even though the problem is discrete.

In the setting of Ising models, however, we do not have a fast parallel algorithm for counting. Several works \cite{feng2021distributed,liu2022simple}
have studied the problem assuming the associated Dobrushin influence matrix has bounded norm. 
By using simultaneous updates, \cite{liu2022simple} obtains a factor-$\fc{n}{C}$ speedup for distributions whose Dobrushin influence matrix has norm bounded by $C$, in particular giving $\mathsf{RNC}$ algorithms when $C=O(1)$ and the mixing time is $O(n\ln n)$. 
The result of \cite{liu2022simple} can also give \Cref{t:main} with a different algorithm, but cannot be used to derive \Cref{t:main-p-spin}. See \Cref{s:comp} for details.

On the practical side, designers of Markov chain Monte Carlo algorithms in discrete spaces have taken inspiration from continuous algorithms, for example, by using gradient information to inform the proposal distribution and allow updating multiple coordinates at once \cite{grathwohl2021oops,zhang2022langevin,rhodes2022enhanced}. Theoretical guarantees for these algorithms remain to be understood.

\subsubsection{Diffusion models and the proximal sampler}

Stochastic localization \cite{eldan2013thin} is a measure-valued stochastic process that converges to a point mass, which is distributed according to a desired distribution $\mu$. As a technique, it gives a way of decomposing probability distributions that has been useful in proving functional inequalities and mixing time~\cite{chen2021almost,chen2022localization}, and more recently, in constructing new, time-\emph{inhomogeneous} algorithms for sampling~\cite{el2022sampling,montanari2023posterior}.

Diffusion models~\cite{sohl2015deep,song2019generative,song2020score} are a successful paradigm for generative modeling in machine learning, where the task is to learn and then generate samples from a distribution where only samples are given. Though the details may differ, they consist of a forward process which adds noise to the data; reversing the process can then generate a sample from random noise. It has been observed~\cite{montanari2023sampling} that a stochastic localization process can be viewed as the reverse process of a diffusion model.

Our analysis of $k$-Glauber dynamics is inspired by the analysis of the proximal sampler \cite{lee2021structured,chen2022improved,fan2023improved},
which does alternating Gibbs sampling by adding Gaussian noise to the current sample, and then ``de-noising" by sampling from the posterior distribution; this fits in the framework discussed above. In their analysis, \cite{chen2022improved} show that proximal sampler mixes at least as fast as Langevin in terms of $\chi^2$ and KL-divergence. \cite{fan2023improved} show a $\wt O(n^{1/2})$ dimension dependence using a carefully chosen Gaussian proposal distribution to implement the posterior sampling step.
We view the $k$-Glauber dynamics as a discrete analogue of the proximal Langevin algorithm, where the noise consists of erasing $k$ coordinates, and our proof follows this analogy. In our application, we also require a careful choice of product distribution for the proposal.

\section{Preliminaries}
\label{s:prelim}
While many of the notions are generalizable, we will restrict ourselves to finite state spaces, and identify all measures with their probability mass functions. For more background on Markov chains, see \cite{montenegro2006mathematical}.

\subsection{Markov kernels}

For finite sets $A$ and $B$, a Markov kernel $K$ from $A$ to $B$ 
is a function $A\times B\to \R_{\ge 0}$ or equivalently, 
a matrix $\R_{\ge 0}^{A\times B}$, where the rows sum to 1. If $\mu$ is a measure on $A$, then $\mu K$ is a measure on $B$; if $f$ is a function $B\to \R$, then $Kf$ is a function $A\to \R$; these correspond to matrix-vector multiplication. Composition of kernels $K_1$ from $A$ to $B$ and $K_2$ from $B$ to $C$ gives a kernel $K_1K_2$ from $A$ to $C$, which corresponds to matrix multiplication. \blu{For $f,g$ functions on $A$ and $\mu$ a measure on $A$, let $\an{f,g}_\mu = \sum_{x\in A}\mu(x)f(x)g(x)$.
For a kernel $K:A\times B\to \R_{\ge 0}$, given measures $\mu_1,\mu_2$ on $A$ and $B$ respectively, we think of $K$ as 
a linear map $L^2(\mu_1)\to L^2(\mu_2)$; then its adjoint $K^*:B\times A\to \R_{\ge 0}$ is a linear map $L^2(\mu_2)\to L^2(\mu_1)$ satisfying $\an{f,Kg}_{\mu_1} = \an{K^*f,g}_{\mu_2}$ for any $f\in L^2(\mu_1)$, $g\in L^2(\mu_2)$.}

\begin{df}
    \vocab{$k$-Glauber dynamics} with stationary distribution $\mu$ on $\Om$ is the Markov chain where at each step, if the current sample is $x$, we choose a subset $S$ uniformly at random in $\binom{\Om}{k}$ (subsets of size $k$), and resample the coordinates in $S$ according to $\mu(X_S|X_{S^c}=x_{S^c})$. Let $P_{\mu,k}$ denote the transition operator. For $k=1$, we simply call it Glauber dynamics, and let $P_\mu$ denote the Markov kernel.
\end{df}

\begin{df}
Let $0\le \ell \le k\le n$. 
Let $\mu$ be a distribution on $\binom{[n]}{k}$.
Define the \vocab{down operator} $\Dn{k}{\ell}$ and \vocab{up operator} $\Up{\ell}{k}$ as Markov kernels $\binom{[n]}{k}\times \binom{[n]}{\ell} \to \R_{\ge 0}$ and $\binom{[n]}{\ell}\times \binom{[n]}{k} \to \R_{\ge 0}$, respectively, with
\begin{align*}
    \Dn{k}{\ell}(A,B) &= \one_{B\subeq A} \rc{\binom k\ell} & 
    \Up{\ell}{k}(B,A) &= \one_{B\subeq A} \fc{\mu(A)}{\sum_{A'\supeq B}\mu(A')}. 
\end{align*}
Let $\mu_\ell=\mu \Dn{k}{\ell}$ for $0\le \ell \le k$, and define the \vocab{$k\lra \ell$ down-up walk} and \vocab{$\ell\lra k$ up-down walk} by
\begin{align*}
\Pdul{k\lra \ell} &= \Dn{k}{\ell}\Up{\ell}{k}&
\Pudl{\ell \lra k} &= \Up{\ell}{k}\Dn{k}{\ell}.
\end{align*}
\end{df}
Note that $\Dn{k}{\ell}$ does not depend on $\mu$ while $\Up{\ell}{k}$ does; we suppress the dependency in the notation. 
Note that $\Dn{k}{\ell}\Dn{\ell}m = \Dn km$ and $\Up{m}{\ell} \Up{\ell}k = \Up{m}{k}$. 
As operators, $\Dn k{\ell}:L^2(\mu_\ell)\to L^2(\mu_k)$ and $\Up{\ell}k:L^2(\mu_k) \to L^2(\mu_\ell)$ are adjoint.

\begin{df}\label{d:hom}
    Let $\mu$ be a measure on $\Om'=\Om_1'\times \cdots \times \Om_n'$. Define the \vocab{homogenization} of $\mu$ to be the measure $\mu^\hom$ over $\binom{\Om}n$, where $\Om = \bigcup_{i=1}^n \Om_i'\times \{i\}$ and $\si\in \Om'$ is identified with $\{(\si_1,1),\ldots, (\si_n, n)\}$. (For short, we will write $\Om = \bigsqcup_{i=1}^n \Om_i'$ in the following.)
    For any property $\cal P$ defined for measures $\binom{\Om}{n}$, we say that $\mu$ satisfies $\cal P$ if $\mu^\hom$ satisfies $\cal P$.
\end{df}
Under this identification, $k$-Glauber dynamics corresponds to the $n\lra n-k$ down-up walk, as the down step corresponds to erasing $k$ coordinates and the up step corresponds to restoring them with the correct conditional probabilities.

\subsection{Functional inequalities}
\label{s:fi}
\begin{df}
    Let $M=(\Om, P)$ be an ergodic, reversible Markov chain with stationary distribution $\mu$. Define the associated Dirichlet form as the inner product
    \[
\sE_P(f,g) = \an{f, (I-P)g}_\mu = \rc 2 \sum_{x,y\in \Om} \mu(x) P(x,y) (f(x)-f(y))(g(x)-g(y)) 
    \]
    When $\mu$ is a distribution on $\Om=\prodo in \Om_i'$, we write $\sE_\mu = \sE_{P_\mu}$; we will similarly make other such replacements without comment.
\end{df}

\begin{df}
    Keeping the assumptions above, we say that $P$ satisfies a \vocab{Poincar\'e inequality} with constant $C$ if for all $f:\Om\to \R$, 
    \[
\Var_\mu(f)\le 
C\sE_P(f,f).
    \]
    We say $\mu$ satisfies a Poincar\'e inequality with constant $C$ if the Glauber dynamics with stationary distribution $\mu$, $P_\mu$, satisfies a Poincar\'e inequality with constant $C$.
\end{df}
When $P$ is self-adjoint ($M$ is reversible), this is the same as saying that $\la_2(P)\le 1-\rc{C}$, where $\la_k(\cdot)$ denotes the $k$th largest eigenvalue.

\begin{df}
    Let $f:\R_{\ge 0}\to \R_{\ge0}$ be a strictly convex function with $f(1)=0$. 
    For measures $\nu\ll \mu$ on $\Om$, define the \vocab{$f$-divergence} by
    \[
\Df(\nu\|\mu) = \E_{x\sim \mu} f\pf{\nu(x)}{\mu(x)}.
    \]
    In particular, define the $\chi^2$ and KL-divergences by $\chis = D_{(x-1)^2}$ and $\KL = D_{x\ln x}$. 
\end{df}

\begin{df}
    We say that Markov kernel $P: \Om_1\times \Om_2\to \R$ satisfies \vocab{$\rh$-contraction in $f$-divergence} with respect to $\mu_1$ if for all $\nu_1\ll\mu_1$, 
    \[
\Df(\nu_1 P\|\mu_1P) \le \rh \Df(\nu_1 \|\mu_1).
    \]
\end{df}
Contraction in $\chi^2$ and KL-divergence is also referred to as variance or entropy contraction, respectively.

\begin{pr}\label{p:chis-pi}
Let $P:\Om_1\times \Om_2\to \R_{\ge 0}$ be a Markov kernel. 
The following are equivalent, for $C\le 1$:
\begin{enumerate}
    \item $P$ satisfies $(1-C)^2$-contraction in $\chi^2$-divergence with respect to $\mu$.
    \item For all $f:\Om_1\to \R$, 
    \[
\Var_{\mu P} (Pf) \le (1-C)^2 \Var_\mu(f).
    \]
    \item (For $\Om_1=\Om_2$, $P$ reversible) $P$ satisfies a Poincar\'e inequality with constant $\rc{C}$.
    \item (For $P$ of the form $P=DD^*$, e.g., $\Pdul{k\lra k-1} = \Dn{k}{k-1}\Up{k-1}k$) $D$ satisfies $(1-C)$-contraction in $\chi^2$-divergence.
    \item (For $P=DD^*)$ $D^*$ satisfies $(1-C)$-contraction in $\chi^2$-divergence. 
\end{enumerate}
\blu{Here, the adjoint is with respect to the measures $\mu$ and $\mu D$.}
\end{pr}

\begin{proof}[Proof sketch.]
We note $\chi^2(\nu\|\mu) = \Var_\mu \pf{\nu}{\mu}$ and $\chi^2(\nu P\|\mu) = \Var_\mu\pf{\nu P}{\mu} = \Var_\mu\pa{\pf{\nu}{\mu}P^*}$ so taking $f=\fc{\nu}{\mu}$, (2) implies (1). 
For any $f$ with $\E_\mu f=0$, we note that for small enough $\ep>0$, $\chi^2(\mu(1+\ep f)P\|\mu) = \ep^2 \Var_{\mu P}(P^*f)$, so (1) implies (2). 
If $\Om_1=\Om_2$ and $P$ is reversible, then $P=P^*$ and (2) is equivalent to $\an{Pf,Pf}_\mu \le (1-C)^2 \an{f,f}$ for all $f\in \one^\perp$, which is equivalent to $\la_2(P)^2 \le (1-C)^2$ and (3).
For (4) and (5), if $P=DD^*$, considering $D$ as an operator $L^2(\mu D)\to L^2(\mu)$, 
note that the squares of the singular values of $D$ (and $D^*$) are the eigenvales of $P$ (ignoring zero eigenvalues), so $\la_2(P) = \si_2(D)^2 = \si_2(D^*)^2$, where $\si_k(\cdot)$ denotes the $k$th largest singular value.
\end{proof}





\begin{df}
A measure $\mu$ on $\binom{[n]}{k}$ satisfies \vocab{$C$-approximate tensorization of entropy} if $\Dn{k}{k-1}$ satisfies $\pa{1-\rc{Ck}}$-contraction in KL-divergence, i.e., for any $\nu\ll \mu$, 
\[
\KL(\nu \Dn{k}{k-1}\| \mu \Dn{k}{k-1})\le \pa{1-\rc{Ck}}\KL(\nu\|\mu).
\]
\end{df}
We have the following alternate characterization for a measure defined on a product space. Define the entropy of a function $f$ on a probability space by $\Ent_\mu[f] = \E_\mu[f\ln f] - \E_\mu[f] \ln \E_\mu[f]$.
\begin{pr}
[{\cite[Lemma 2.7]{chen2021optimal}}]
\label{p:atoe-equiv}
    Let $\mu$ be a measure on $\Om=\Om_1'\times \cdots \times \Om_n'$. Then $\mu$ satisfies $C$-approximate tensorization of entropy iff for all $f:\Om\to \R_{\ge 0}$, 
    \[
\Ent_\mu[f]\le C \sumo kn \E_\mu \ba{\Ent_{\mu(X_k=\cdot|X_{\sim k} = x_{\sim k})}[f]},
    \]
    where $\sim k$ denotes the coordinates besides $k$.
\end{pr}

\begin{rem}\label{r:chi-vs-kl}
\cref{p:chis-pi} shows that for contraction in $\chi^2$-divergence, nothing is lost if we consider $\Pdul{k\lra k-1}$ or $\Dn{k}{k-1}$, $\Up{k-1}k$ separately. 
However, the distinction is important for KL, as contraction of $\Pdul{k\lra k-1}$ may not imply contraction of $\Dn{k}{k-1}$ or $\Up{k-1}k$  separately; hence the definition of approximate tensorization of entropy. Approximate tensorization of entropy is stronger than the modified log-Sobolev inequality (which implies mixing for $\Pdul{k\lra k-1}$), but weaker than the log-Sobolev inequality.
\end{rem}

\subsection{Additional notation}

For $f:\prod_{i=1}^n \Om_i'\to \R$, and $x\in \prod_{i\in S^c} \Om_i'$, define the restriction $f_x:\prod_{i\in S} \Om_i'$ by $f_x(y) = f(x,y)$ with $(x,y)$ treated as an element of $\prod_{i=1}^n \Om_i'$. 

Let $\seti xib$ denote $x$ with $x_i$ set to $b$. 
For $f:\{\pm 1\}^n$, 
let
$\D_i f (x):= \rc 2[f(\seti xi1) - f(\seti xi{-1})]$ and define $\nb f:\Bn \to \R^n$ by
\[
\nb f(x) = (\D_1f(x),\ldots, \D_nf(x))
\]
and $\nb^2 f:\Bn \to \R^{n\times n}$ by $(\nb^2 f(x))_{i,j} = \D_i\D_j f(x)$ (note $(\nb^2 f(x))_{i,i} = 0$).

For $x\in \Bn$ and $S\subeq [n]$, let $x^S$ denote $\prod_{i\in S} x_i$. 
For a function $f:\Bn\to \R$, we denote the degree $d$ part of $f$ by $f^{(d)}$, and define $f^{\ge d} = \sum_{p\ge d} f^{(p)}$, etc., so that we have the decomposition
\[
f(x) = \sum_{p=0}^n
f^{(p)}(x) = \sum_{p=0}^n \sum_{|I|=p} a_I x^I
\]
for some coefficients $a_I$. 
We take $f_x^{(d)}$ to mean that we take the restriction first and then the degree-$d$ part.

For a scalar-valued function $f$ and $x\in \R^n$, we let $f(x)$ denote coordinate-wise evaluation.

\section{$k$-Glauber mixes $k$ times as fast}

To show that $k$-Glauber mixes $k$ times more quickly than Glauber dynamics, write $P_\mu = \Dn n{n-1}\Up{n-1}n$ and $P_{\mu,k} = \Dn{n}{n-1}\cdots \Dn{n-k+1}{n-k}\Up{n-k}{n-k+1}\cdots \Up{n-1}n$; the task is then to show that $\Dn{j}{j-1}$, $j\le n$ are roughly at least as contractive as $\Dn n{n-1}$. We note that this is like the reverse of the usual ``local-to-global" argument for high-dimensional expanders~\cite{alev2020improved} (and will be easier!). We also note this approach relates to inductive arguments in prior work (e.g., \cite[Lemma 11]{cryan2019modified}).

Viewing the problem in this way, we note the similarity to the proximal sampler \cite{chen2022improved}, each step of which involves adding and removing Gaussian noise; they bound the contraction in $\chi^2$ and KL-divergence of this process based on the Poincar\'e and log-Sobolev constants of the original distribution.

As inspiration, we first recall the following fact, which bounds the Poincar\'e or log-Sobolev constant of a convolution of two measures on $\R^n$. (The convolution $\mu_1*\mu_2$ is defined as the distribution of $X+Y$ where $X\sim \mu_1$ and $Y\sim \mu_2$ are independent.) As we will not cover the theory of functional inequalities over $\R^n$, this is meant only as a suggestion of how we might proceed. 
For details and generalizations, see \cite{chafai2004entropies}.
\begin{lem*} 
    Suppose that $\mu_1,\mu_2$ are distributions on $\R^n$ with Poincar\'e constants $C_1, C_2$, respectively. Then $\mu_1*\mu_2$ has Poincar\'e constant bounded by $C_1+C_2$. The same holds true for the log-Sobolev constant.
    In particular, this holds true for $\mu_2$ being a Gaussian of variance $C_1$.
\end{lem*}
\begin{proof}[Proof sketch.]
    Let $M_m$ denote multiplication by $m$. Then $\mu_iM_{m_i}^{-1}$ has Poincar\'e constant $m_i^2C_i$. By tensorization, $\mu = \mu_1M_{m_1}^{-1} \ot \mu_2M_{m_2}^{-1}$ has Poincar\'e constant $\max \bc{C_1m_1^2, C_2m_2^2}$. Consider the projection $\pi(x_1,x_2)=\fc{x_1}{m_1} + \fc{x_2}{m_2}$. 
    When $\rc{m_1^2}+\rc{m_2^2}=1$, 
    this can be realized as projection onto the vector $(\rc{m_1},\rc{m_2})$, so the Poincar\'e constant does not increase: $\CP(\mu\pi^{-1})\le \CP(\mu)$. Note that we exactly have $\mu\pi^{-1} = \mu_1*\mu_2$. 
    Choosing $\rc{m_1^2} = \fc{C_1}{C_1+C_2}$
    and 
    $\rc{m_2^2}= \fc{C_2}{C_1+C_2}$ gives the bound. The same argument works for the log-Sobolev constant.
\end{proof}
We will carry out the same steps as in this proof: define a tensorization operation which preserves the Poincar\'e or approximate tensorization of entropy constant, and a projection which can only improve it.
As in the sketch above, it will help to weight the components appropriately in the tensorization step. 

To obtain the $\Dn{k}{k-1}$ operator from the $\Dn{n}{n-1}$ operator, we will tensorize with the appropriate ``noise" distribution, which in this case is that of erasing $n-k$ coordiates. We note that while bit-flip noise is the more natural analogue of gaussian noise, the ``denoising" step is more difficult to implement, while erasure noise connects more nicely with existing notions.


\subsection{Tensorization and projection}

The following proposition is similar to classical results on preservation of functional inequalities under tensorization, which corresponds to taking a product of Markov chains. However, we need to work with operators between different spaces to obtain contraction for the down operator, so we need the result given in \cref{l:tens} . As in Remark~\ref{r:chi-vs-kl}, nothing would be lost for $\chi^2$-divergence if we considered the down-up walk---so we could use existing results---but for KL we need to bound contraction for just the down operator. This is important because we aim to bound contraction for $\Pdul{n\lra n-k} = \Dn{n}{n-1}\cdots \Dn{n-k+1}{n-k} \Up{n-k}{n-k+1}\cdots \Up{n-1}{n}$, and in the case of KL, we cannot obtain this by bounding contraction for just operators of the form $\Pdul{k\lra k-1}=\Dn{k}{k-1}\Up{k-1}k$.

\begin{pr}[Contraction under tensorization]
\label{l:tens}
Suppose that $P_i$ is a Markov kernel from $\Om_i$ to $\Om_i'$, for $i=1, 2$. Let $I_i$ denote the identity kernel on $\Om_i$ and define $P=p(P_1\ot I_2) + (1-p)(I_1\ot P_2)$ as a Markov kernel from $\Om_1\times \Om_2$ to $(\Om_1'\times \Om_2) \sqcup (\Om_1 \times \Om_2')$. 

Let $\Df = \KL$ or $\chis$. 
If $P_i$ satisfies $(1-\ka_i)$-contraction in $f$-divergence with respect to $\mu_i$, then $P$ satisfies $(1-\ka)$-contraction in $f$-divergence with respect to $\mu_1\ot\mu_2$, where 
$\ka = \min\{p\ka_1,(1-p)\ka_2 \}$.
In particular, if $p=\fc{\ka_2}{\ka_1+\ka_2}$, then $\ka=\fc{\ka_1\ka_2}{\ka_1+\ka_2} = \rc{\ka_1^{-1}+\ka_2^{-1}}$.
\end{pr}
\begin{proof}
First consider KL-divergence. 
Let $\nu$ be a measure on $\Om_1\times \Om_2$. Let $\Pi_i$ and $\Pi_i'$ denote the projection kernels to $\Om_i$ and $\Om_i'$ respectively, for $i=1, 2$ (i.e., taking marginals). 
    We calculate
    \begin{align}
    \nonumber
        B:&=\KL(p\nu (P_1\ot I_2) + (1-p)\nu (I_1\ot P_2)\| p(\mu_1'\times \mu_2) + (1-p) (\mu_1\times \mu_2'))\\
        &= 
        p \KL (\nu (P_1\ot I_2)\|\mu_1'\times \mu_2) + 
        (1-p) \KL (\mu (I_1 \ot P_2)\|\mu_1\times \mu_2'),
        \label{e:kl-tens}
    \end{align}
    since the spaces $\Om_1'\times \Om_2$ and $\Om_1\times \Om_2'$ are disjoint. 
    Now, by the chain rule of KL-divergence and entropy contraction, 
    \begin{align*}
        \KL (\nu (P_1\ot I_2)\|\mu_1'\times \mu_2)
        &= \E_{x_2\sim \mu_2} \KL((\nu(P_1\ot I_2))(\cdot |x_2)\| \mu_1')
        + \KL(\nu(P_1\ot I_2)\Pi_2\|\mu_2)\\
        &= \E_{x_2\sim \mu_2} \KL(\nu(\cdot |x_2) P_1 \| \mu_1')
        + \KL(\nu\Pi_2 \|\mu_2)\\
        &\le (1-\ka_1) \E_{x_2\sim \mu_2} 
        \KL(\nu(\cdot |x_2)\|\mu_1) + 
        \KL(\nu\Pi_2 \|\mu_2).
    \end{align*}
    By symmetry, the same calculation holds for the second term in~\eqref{e:kl-tens}, giving us
    \begin{align*}
        B &\le 
        p\ba{(1-\ka_1) \E_{x_2\sim \mu_2} 
        \KL(\nu(\cdot |x_2)\|\mu_1) + 
        \KL(\nu\Pi_2 \|\mu_2)}\\
        &\quad + 
        (1-p) \ba{(1-\ka_2) \E_{x_1\sim \mu_1} 
        \KL(\nu(\cdot |x_1)\|\mu_2) + 
        \KL(\nu\Pi_1 \|\mu_1)} .
    \end{align*}
    We wish to compare this with
    \begin{align*}
        A:&= \KL(\mu \|\mu_1\times \mu_2)\\
        &= \ub{\E_{x_2\sim \mu_2} \KL(\nu(\cdot|x_2)\|\mu_1)}{C_{2,1}} + 
        \ub{\KL(\nu\Pi_2 \|\mu_2)}{C_2}\\
        &= \ub{\E_{x_1\sim \mu_1}
        \KL(\nu(\cdot|x_1)\|\mu_2)}{C_{1,2}} + 
        \ub{\KL(\nu\Pi_1 \|\mu_1)}{C_1}.
    \end{align*}
    By convexity of KL-divergence, $C_{1,2}\ge C_2$ and $C_{2,1}\ge C_1$. Hence
    \begin{align*}
        B&\le p((1-\ka_1) C_{2,1} + C_2) + (1-p)((1-\ka_2) C_{1,2} + C_1) \\
        &\le 
        \max\{(1-p)(1-\ka_2) + p, p(1-\ka_1) + (1-p)\}  [p(C_{2,1} + C_2) + (1-p) (C_{1,2}+C_1)] \\
        & = \pa{1-\min\{(1-p)\ka_2, p\ka_1\}} A. 
    \end{align*}
    
Next consider $\chi^2$-divergence. It suffices to prove the following equivalent statement on contraction of variance (\cref{p:chis-pi}): for any $f:\Om_1\times \Om_2\to \R$, considering $P^*=p(P_1^*\ot I_2) \opl (1-p) (I_1\ot P_2^*)$, we have 
\[
\Var_{\mu P} (P^*f) \le (1-\ka) \Var_\mu(f),
\]
which is equivalent to $\si_2(P^*)\le 1-\ka$ and hence to $\la_2(PP^*)^2 = \si_2(PP^*)^2\le (1-\ka)^2$.
Now $PP^*= p(P_1P_1^*\ot I_2) + (1-p)(I_1\ot P_2P_2^*)$ is exactly the transition matrix of the weighted product of two Markov chains with $\la_2(P_iP_i^*)\le 1-\ka_i$, so it is well-known that 
\[
\la_2(PP^*) \le \max\{p+(1-p)(1-\ka_2), (1-p)+p(1-\ka_1)\} = 1-\max\{(1-p)\ka_2, p\ka_1\}.
\]
(A quick way to see this is as follows: if $\{f_i\}$ are the eigenvectors of $P_1P_1^*$ and $\{g_j\}$ are the eigenvectors of $P_2P_2^*$, then $\{f_ig_j\}$ are the eigenvectors of $P^*P$, and the second largest eigenvalue is when $f_i=1$ or $g_i=1$.)
\end{proof}
Though we do not do it here, the proofs can be put on the same footing and generalized using the notion of $f$-entropy \cite{chafai2004entropies}. 

\begin{pr}[Contraction under projection]
\label{l:proj}
Suppose that $P$ is a Markov kernel from $\Om_1$ to $\Om_2$, and $\mu_1,\mu_2$ are measures on $\Om_1,\Om_2$ such that $\mu_1P=\mu_2$. Let $\pi_i:\Om_i\to \Om_i'$ be maps and $\mu_i' = \mu_i\pi_i^{-1}$. Define the projected Markov kernel $P':\Om_1'\times \Om_2'\to \R_{\ge0}$ by 
\[
P'(x_1',x_2') = 
\sumr{x_1\in \Om_1}{\pi_1(x_1)=x_1'}
\sumr{x_2\in \Om_2}{\pi_2(x_2)=x_2'}
\mu_1(x_1|\pi(x_1)=x_1') P(x_1,x_2).
\]
If $P$ satisfies $\rh$-contraction in $f$-divergence with respect to $\mu_1$, then $P'$ also satisfies $\rh$-contraction in $f$-divergence with respect to $\mu_1'$. 
\end{pr}
In words, in the projected Markov chain, given $x_1'$, we draw $x_1$ projecting to $x_1'$ from the ``prior", move according to $P$, and then project back down; then the  projected kernel always has at least as much contraction as the original one. See e.g., \cite{madras2002markov} for the statement for the Poincar\'e constant.
\begin{proof}
    Let $\nu_1'\ll \mu_1'$ be a measure on $\Om_1'$. Define $\nu_1(x) = \nu_1'(\pi_1(x)) \mu_1(X=x|\pi_1(X)=\pi_1(x))$. Then
    $\E_{\mu_1(\cdot |\pi_1(X)=x_1')}f\pf{\nu_1(x)}{\mu_1(x)} = \fc{\nu_1'(x_1')}{\mu_1'(x_1')}$, so 
    \[
\Df(\nu_1'\|\mu_1') = 
\E_{x_1'\sim \mu_1'} f\pf{\nu_1'(x_1')}{\mu_1'(x_1')} =
\Df(\nu_1\|\mu_1).
    \]
    By contraction of $P$ and the data processing inequality, \[\Df((\nu_1P)\pi_2^{-1}\|\mu_2') \le \Df(\nu_1P\|\mu_2)\le \rh \Df(\nu_1\|\mu_1)
    = \rh \Df(\nu_1'\|\mu_1').\] Finally, note that $(\nu_1P)\pi_2^{-1} = \nu_1'P'$ by definition of $P'$. 
\end{proof}

\subsection{Contraction improves going down} 

We now show that for $k<n$, contraction in KL and $\chi^2$ for $\Dn{k}{k-1}$ will only be better than contraction of $\Dn{n}{n-1}$, except up to an additive constant.

\begin{lem}\label{l:bl}
    Let $\mu$ be the uniform distribution on $\binom{[n]}{k}$. Then $\Dn{k}{k-1}$ has $\pa{1-\rc{k}}$-contraction in KL, and $\mu$ satisfies 1-approximate tensorization of entropy.
    Moreover, $\Dn{k}{k-1}$ and $\Up{k-1}{k}$ satisfy $\pa{1-\fc{n}{k(n-k+1)}}$-contraction in $\chi^2$.
\end{lem}
We note that the down-up walk on the uniform distribution on $\binom{[n]}{k}$ is a rescaling of the Bernoulli-Laplace diffusion model, for which mixing and functional inequalities have been extensively studied~\cite{diaconis1987time,lee1998logarithmic,filmus2022log,bristiel2021entropy,salez2021sharp}, and we have not attempted to find the best result for KL.
\begin{proof}
Note that $\mu$ is log-concave, by the results of~\cite{anari2018log} and the fact that $\binom{[n]}{k}$ is a matroid. 
Then 1-approximate tensorization of entropy follows from~\cite[Theorem 5]{anari2021entropic}. 

To note the improved bound for $\chi^2$, note first that the probability of staying at the same set is $\rc{n-k+1}$, so $\Pdul{k\lra k-1} = \rc{n-k+1}I + \fc{n-k}{n-k+1} \Pdun{k\lra k-1}$ where $\Pdun{k\lra k-1}$ is the ``non-lazy" walk which swaps an occupied and non-occupied space at random. Hence $I-\Pdul{k\lra k-1} =\fc{n-k}{n-k+1}(I-\Pdun{k\lra k-1})$. 
This in turn satisfies
\[
I - \Pdun{k\lra k-1} = \fc{n(n-1)/2}{k(n-k)} \cdot \fc{2}{n-1} \cdot L
\]
where $L$ is the generator of the Bernoulli-Laplace diffusion model with parameters $(n,k)$ (a random transposition occurs with rate 1). Here, the $\fc{2}{n-1}$ is the scaling factor to convert to a discrete-time walk and $\fc{n(n-1)/2}{k(n-k)}$ takes into account that we are randomly choosing from $k(n-k)$ transpositions that move a particle to an unoccupied space, rather than an arbitrary transposition. By~\cite{salez2021sharp}, $\la_{\min}(L)=1$. Hence the spectral gap for $\Pdul{k\lra k-1}$ is 
\[
\fc{n-k}{n-k+1} \cdot 
\fc{n(n-1)/2}{k(n-k)} \cdot \fc{2}{n-1} = \fc{n}{k(n-k+1)}.\qedhere 
\]
\end{proof}
\begin{lem}\label{l:mono}
Let $\mu$ be a distribution on $\prodo in \Om_i'$ and $\mu^{\hom}$ its homogenization on $\binom{\Om}n$, where $\Om = \bigsqcup_{i=1}^n \Om_i'$, and $\mu_m = \mu^{\hom}\Dn{n}{m}$. Let $\Df = \KL$ or $\chis$. For each $k$, let $\ka_{k}$ be the largest number such that $D_{k}:=\Dn{k}{k-1}$ satisfies $(1-\ka_{k})$-contraction in $f$-divergence with respect to $\mu_{k}$. Suppose that for the uniform distribution $\binom{[n]}{k}$ that $\Dn k{k-1}$ satisfies $(1-\kabl{k})$-contraction in $f$-divergence.
Then 
\[
\ka_{k} \ge 
\fc{n \ka_n \kabl{k}}{n\ka_n + k\kabl{k}}
\ge\begin{cases}
\fc{n\ka_n}{k((n-k+1) \ka_n + 1)}, & \Df = \chis\\
\fc{n\ka_n}{k(n \ka_n+1)}, & \Df=\KL.
\end{cases}
\]
\end{lem}
\begin{proof}
    Consider the kernel $P = p(D_n \ot I_{\binom{[n]}{k}}) + (1-p) (I_{\binom{\Om}{n}}\ot D_{k})$ from $\binom{\Om}{n}\times \binom{[n]}{k}$ to $\binom{\Om}{n-1}\times \binom{[n]}{k} \cup \binom{\Om}{n} \times \binom{[n]}{k-1}$, where $D_k$ denotes the down operator $\binom{[n]}{k}\times \binom{[n]}{k-1}\to \R_{\ge0}$. 
    By tensorization (\cref{l:tens}), for $p=\fc{\kabl{k}}{\ka_n + \kabl{k}}$ and  $\ka=\rc{\ka_n^{-1} + \kabl{k}^{-1}}$, $P$ satisfies $(1-\ka)$-contraction in $f$-divergence with respect to $\mu_n\ot \mathsf{Uniform}\binom{[n]}{k}$.  
    Define the projections 
    $\pi_1:\binom{\Om}{n}\times \binom{[n]}{k} \to \binom{\Om}k$ and 
    $\pi_2:\binom{\Om}{n-1}\times \binom{[n]}{k} \cup \binom{\Om}{n} \times \binom{[n]}{k-1}\to \binom{\Om}{k-1} \cup \binom{\Om}{k}$
    both as 
    \begin{align*}
        \pi(S, A) & = S\cap \bigcup_{i\in A}\Om_i,
    \end{align*}
    i.e., if $S$ corresponds to $x\in \prodo in \Om_i'$, we keep only the coordinates in the set $A\in \binom{[n]}{k}$. 
    
    We claim the projected kernel as defined in \cref{l:proj} is 
    \begin{align*}
        P' &= \ub{\pa{1-\fc{p(n-k)}n}D_{k}}{\mathrm{(I)}} + \ub{\fc{p(n-k)}{n}I}{\mathrm{(II)}}
    \end{align*}
    from $\binom{\Om}{k}$ to $\binom{\Om}{k-1}\cup \binom{\Om}{k}$. To see this, we first identify $x_A\in \prodo_{i\in A}\Om_i'$ with its homogenization in $\binom{\Om}{|A|}$ (\cref{d:hom}). Then $\mu_1(\cdot|\pi_1(x_1)=x_A)$ is the distribution of $(x,A)$ where $x_{A^c}$ is distributed as $\mu(X_{A^c}=x_{A^c}|X_A=x_A)$. Under the transition $P$:
    \begin{enumerate}
        \item With probability $p$, we remove a coordinate $i$ of $x$. In this case,
        \begin{enumerate}
            \item with probability $\fc{k}n$, $i\in A$, and projection by $\pi_2$ then gives $x_{A\bs \{i\}}$ for a random index $i$.
            \item with probability $\fc{n-k}n$, $i\nin A$, and projection by $\pi_2$.
        \end{enumerate}
        \item With probability $1-p$, we remove an element $i$ of $A$, and projection gives $x_{A\bs \{i\}}$ for a random $i\in A$.
    \end{enumerate}
    Then (1a) and (2) give the term (I) and (1b) gives the term (II).

    Because 
    \begin{align*}
&\Df\pa{\nu P'\|\pa{1-\fc{p(n-k)}n} \mu_{k-1} + \fc{p(n-k)}n \mu_{k}} \\
&= \pa{1-\fc{p(n-k)}n}\KL(\nu D_{k}\| \mu_{k-1}) + 
\fc{p(n-k)}n\KL(\nu \| \mu_{k}) 
    \end{align*}
    (the component measures have disjoint support), 
    we have that $D_{k}$ satisfies $(1-\ka_{k})$-contraction in $f$-divergence iff $P'$ satisfies $\pa{1-\ka_{k}\pa{1-\fc{p(n-k)}n}}$-contraction in $f$-divergence.
    \cref{l:proj} then gives us
    \begin{align*}
        \ka_{k}\pa{1-\fc{\kabl{k}}{\ka_n + \kabl{k}}\cdot \fc{n-k}n} &\ge \rc{\ka_n^{-1} + \kabl{k}^{-1}}\\
        \implies 
        \ka_{k} &\ge 
        \fc{n \ka_n \kabl{k}}{n\ka_n + k\kabl{k}}.
    \end{align*}
    Plugging in the result of Lemma~\ref{l:bl} then gives the result.
\end{proof}

\begin{thm}\label{t:du}
Let $\mu$ be a distribution on $\Om = \prodo in \Om_i'$, and consider the down and up operators defined with respect to its homogenization $\mu^\hom$. 
\begin{enumerate}
    \item If $\Dn{n}{n-1}$ satisfies $(1-\rc{Cn})$-contraction in $\chi^2$, 
    then $\Dn k{\ell}$ satisfies $\prod_{j=\ell+1}^k\pa{1-\rc{j\pa{C+\fc{n-j+1}{n}}}}$-contraction in $\chi^2$.
    \item If $\Dn{n}{n-1}$ satisfies $(1-\rc{Cn})$-contraction in KL (i.e., $\mu$ satisfies $C$-approximate tensorization of entropy), then $\Dn k{\ell}$ and $\Pdul{k\lra \ell}$ satisfy $\prod_{j=\ell+1}^k\pa{1-\rc{j(C+1)}}$-contraction in KL.
\end{enumerate}
\end{thm}
We note that our tensorization construction in Lemma~\ref{l:tens} currently depends on $\mu$ being a homogeneous distribution; it would be interesting to extend it beyond this case.
\begin{proof}
If $\Dn{n}{n-1}$ satisfies $\pa{1-\rc{Cn}}$-contraction in $\chi^2$-divergence, then by Lemma~\ref{l:mono}, $\Dn{j}{j-1}$ satisfies $\pa{1-\rc{j\pa{C+\fc{n-j+1}{n}}}}$-contraction in $\chi^2$-divergence. 
If $\Dn{n}{n-1}$ satisfies $\pa{1-\rc{Cn}}$-contraction in $\chi^2$-divergence, then we have $\pa{1-\rc{j(C+1)}}$-contraction in KL-divergence. 
Because $\Dn{k}{\ell} = \Dn{k}{k-1}\cdots \Dn{\ell+1}{\ell}$, taking the product from $\ell+1$ to $k$ gives the result.
\end{proof}

From this, we can conclude that $k$-Glauber dynamics mixes at least $\Om(k)$ times as fast in $\chi^2$-divergence as Glauber dynamics, and assuming approximate tensorization of entropy, mixes at least $\Om(k)$ times as fast in KL-divergence.
\begin{cor}[$k$-Glauber mixes $k$ times as fast]
\label{c:k-glauber}
\label{t:k-glauber}
Let $\mu$ be a distribution on $\Om = \prodo in \Om_i'$. 
\begin{enumerate}
    \item If $\mu$ satisfies a Poincar\'e inequality with constant $Cn$, then 
    \[
    \la_2(P_{\mu,k}) \le \pa{1-\fc{k}{n+1}}^{\fc{1}{C+1}} = 1-\Om\pa{\max\bc{\fc{k}{(C+1)n},1}}
    \]
    and $P_{\mu,k}$ satisfies a Poincar\'e inequality with constant $\Om\pf{(C+1)n}{k}$.
    \item If $\mu$ satisfies $C$-approximate tensorization of entropy, then $P_{\mu,k}$ satisfies contraction in KL-divergence with constant
    \[
    \pa{1-\fc{k}{n+1}}^{\rc{C+1}} = 1-\Om\pa{\max\bc{\fc{k}{(C+1)n},1}}.
    \]
\end{enumerate}
\end{cor}
\begin{proof}
Note that $P_{\mu,k}$ corresponds to $\Pdul{n\lra n-k}$ after homogenization. 
    For variance, we recall \cref{p:chis-pi} which relates the Poincar\'e constant of $P_\mu$ and $\Dn{n}{n-1}$, and the Poincar\'e constant of $P_{\mu,k}$ and $\Dn{n}{n-k}$. 
    The corollary then follows from \cref{t:du} and the calculation
    \begin{align*}
\prod_{j=n-k+1}^n \pa{1-\rc{j(C+1)}} &\le e^{-\rc{C+1}\sum_{j=n-k+1}^n \rc{j}} 
\le e^{-\rc{C+1} \ln\pf{n+1}{n-k+1}}\\
&= \pa{1-\fc{k}{n+1}}^{\rc{C+1}} = 1-\Om\pf{k}{(C+1)n}.\qedhere
    \end{align*}
\end{proof}

Our theorem also implies contraction for $D_{2\to 1}$; this forms a kind of converse to local-to-global arguments that start with contraction of $\Dn{2}{1}$ \cite{alimohammadi2021fractionally}.  
\begin{cor}
Let $\mu$ be a distribution on $\Om = \prodo in \Om_i'$, and consider the down and up operators defined with respect to its homogenization $\mu^\hom$. 
\begin{enumerate}
    \item If $\Dn{n}{n-1}$ satisfies $(1-\rc{Cn})$-contraction in $\chi^2$, 
    then $\Dn 21$ satisfies $\pa{1-\rc{2(C+1)}}$-contraction in $\chi^2$ and $\la_2(\Pudl{1\lra 2}) \le {1-\rc{2(C+1)}}$.
    \item If $\Dn{n}{n-1}$ satisfies $(1-\rc{Cn})$-contraction in KL (i.e., $\mu$ satisfies $C$-approximate tensorization of entropy), then $\Dn 21$ satisfies $\pa{1-\rc{2(C+1)}}$-contraction in KL.
\end{enumerate}
\end{cor}
\begin{proof}
    This follows from substituting $k=2$, $\ell=1$ into \cref{t:du}, appealing to \cref{p:chis-pi} for the $\chi^2$ result.
\end{proof}



\section{Parallel sampling for Ising models}

To apply \cref{t:k-glauber} to the Ising model for $\ve{J}<1$, we use the fact that the Ising model satisfies approximate tensorization of entropy. 
\begin{thm}\label{t:atoe-ising}
    Suppose $\ve{J}<1$. Then $\mu_{J,h}$ satisfies approximate tensorization of entropy with constant $\rc{1-\ve{J}}$. 
\end{thm}

The proof is in \cref{s:ising}. We first introduce a generic guarantee for approximate rejection sampling in \cref{s:ars}, based on establishing concentration for the difference of the log-pdfs. In \cref{s:conc}, we show that using an approximating product distribution---chosen as the solution to a variational problem---is sufficient as a proposal distribution for $\mu_{J,h}$ when $\ve{J}_F$ is small enough, using the Hanson-Wright inequality. We put everything together in \cref{s:proof}, combining the speedup of $k$-Glauber dynamics, known mixing for the Ising model, guarantee on the approximate rejection sampler, with a careful analysis of the recursion in the algorithm.

\subsection{Approximate tensorization of entropy for the Ising model}
\label{s:ising}

For $\ve{J}<1$, we prove approximate tensorization of entropy for the Ising model $\mu_{J,h}$ (\cref{t:atoe-ising}) by using the fact that it holds for rank-1 Ising models and using the 
needle decomposition in \cite{eldan2022spectral}. This is the same way that the modified log-Sobolev inequality is proved in \cite{eldan2022spectral}.

\begin{pr}[{\cite[Proposition 32]{anari2021entropic}}]
\label{p:atoe-needle}
    Suppose $\ve{u}<1$. Then $\mu_{uu^\top,h}$ satisfies approximate tensorization of entropy with constant $\rc{1-\ve{u}^2}$. 
\end{pr}

\begin{thm}[Needle decomposition of Ising measures \cite{eldan2022spectral}]
\label{thm:needle-decomposition}
Consider an Ising model $\mu=\mu_{J,h}$ on $\Bn$ with $J\succeq 0$. Let $f : \{\pm 1\}^n \to \mathbb{R}$ be any function. There exists a mixture decomposition (depending on $f$)
\[ \mu(x) = \int \mu_{u,v}(x) \,d\pi(u,v) \]
where $\pi$ is a probability measure on $\mathbb{R}^{2n}$ such that:
\begin{enumerate}
    \item $\pi$-almost surely, $\mu_{u,v}$ is a probability measure of the form
    \[ \mu_{u,v}(x) = \frac{1}{Z_{u,v}} \exp\pa{\frac{1}{2} \an{u, x}^2 + \an{v, x}}, \]
    i.e., a rank one Ising model (``needle''). Furthermore $\ve{u} \le \ve{J}$.
    \item $\E_{\mu_{u,v}}{f(X)} = \E_{\mu}{f(X)}$ $\pi$-almost surely.
\end{enumerate}
\end{thm}

\begin{proof}[Proof of \cref{t:atoe-ising}]
Let $C=\rc{1-\ve{J}}$. 
We use the needle decomposition (\cref{thm:needle-decomposition}) to decompose the entropy, noting $\E_{\mu_y}[f]$ is constant. Let $y=(u,v)$ in the needle decomposition.
\begin{align}
\nonumber
    \Ent_\mu[f] &= 
    \E_{y\sim \pi} [\Ent_{\mu_y}[f]] + \Ent_{y\sim \pi} [\E_{\mu_y}f]
    = \E_{y\sim \pi} [\Ent_{\mu_y}[f]]\\
    &\le \E_{y\sim \pi} C \sum_{i=1}^n \E_{x\sim \mu_y} [\Ent_{\mu_y(X_i=\cdot|X_{\sim i}=x_{\sim i})}[f]] 
    \label{e:use-atoe}\\
    &= C \E_{x\sim \mu} \sum_{i=1}^n \E_{y\sim \pi(\cdot |X_{\sim i}=x_{\sim i})}[\Ent_{\mu_y(X_i=\cdot|X_{\sim i}=x_{\sim i})}[f]] 
    \label{e:atoe-ising}
\end{align}
In~\eqref{e:use-atoe} we use \cref{p:atoe-needle} and \cref{p:atoe-equiv}. Here, $\pi(\cdot |X_{\sim i}=x_{\sim i})$ denotes the conditional distribution of $y$ given $x_{\sim i}$, when we have $y\sim \pi$ and $x\sim \mu_y$. 
Also by entropy decomposition,
\begin{align*}
    \Ent_{\mu(X_i=\cdot |X_{\sim i}=x_{\sim i})}[f] 
    &= \E_{y\sim \pi(\cdot |X_{\sim i}=x_{\sim i})} [\Ent_{\mu_y(X_i=\cdot|X_{\sim i}=x_{\sim i})}[f]]
    + \Ent_{\pi(\cdot |X_{\sim i}=x_{\sim i})}[\E_{\mu_y(X_i=\cdot|X_{\sim i}=x_{\sim i})} [f]].
\end{align*}
Taking expectation over $x\sim \mu$, 
\begin{align*}
    \E_{x\sim \mu} \Ent_{\mu(X_i=\cdot |X_{\sim i}=x_{\sim i})} [f]
    &\ge \E_{x\sim \mu}\E_{y\sim \pi(\cdot |X_{\sim i}=x_{\sim i})} [\Ent_{\mu_y(X_i=\cdot|X_{\sim i}=x_{\sim i})}[f]]
\end{align*}
Applying this to \eqref{e:atoe-ising} 
gives 
\begin{align*}
    \Ent_\mu[f] &\le C \sumo in \E_{x\sim \mu}[ \Ent_{\mu(X_i=\cdot |X_{\sim i}=x_{\sim i})}[f]].
\end{align*}
By \cref{p:atoe-equiv}, $\mu$ hence satisfies $C$-approximate tensorization of entropy.
\end{proof}

\subsection{Approximate rejection sampling}
\label{s:ars}
Conditional distributions of the Ising model are again Ising models. We will show that with large probability, if we pick a random subset of not-too-large size, then we can approximate the distribution of those coordinates with a product distribution, and hence use the product distribution as a proposal for approximate rejection sampling. We discuss further the choice of product distribution (given in Algorithm~\ref{a:qars}) in Section~\ref{s:conc}.

First, we give a generic guarantee for the rejection sampling \cref{a:ars}, which can be used whenever the log-ratio between desired and proposal distributions $\ln \dd{P}{Q}$ has sufficiently decaying exponential tails. Note that because the normalizing constants are unknown, we draw two samples and use one as a reference to decide whether to accept the other. 
\cref{l:ars} appears as \cite[Lemma 2]{fan2023improved} specialized to the distribution they consider, but holds more generally. We give the proof for completeness.

\begin{algorithm}[h!]
\caption{Quadratic approximate rejection sampler ($\mathsf{QuadraticApproxRejectionSampler}$)}
\begin{algorithmic}[1]
\INPUT Hamiltonian $H$ in the form $H(x) = C + \an{h,x} + \rc 2\an{x,Ax} + H_{\ge 3}(x)$, $\de$ such that $\ve{A^{\invdiameter}}_F\le \de < c_3$ and $0\preceq A \preceq (1-c)I$, error $\ep$. 
\State Let $u_0= \mathbf 0$.
\For{$t$ from 1 to 
$T=\Te\pa{\rc{c} \ln \pf{\sqrt n}{\de}}$}
    \State Let $u_t = A^{\invdiameter}\tanh(h+u_{t-1})$.
\EndFor
\State Let $\wh h = u_T$. 
\OUTPUT $\mathsf{ApproxRejectionSampler}(q(x) \propto \exp\pa{\langle{h+\wh h,x}\rangle, g(x)= H(x) - \langle{h + \wh h, x}\rangle, \pf{2}{\ep}^{\fc{\de}{c_3-\de}}}$ (See Algorithm~\ref{a:ars}.)
\end{algorithmic}
\label{a:qars}
\end{algorithm}

\begin{algorithm}[h!]
\caption{Approximate rejection sampler ($\mathsf{ApproxRejectionSampler}$)}
\begin{algorithmic}[1]
\INPUT Oracle for sampling from $Q$, function $g$ such that $\dd{P}{Q} \propto e^{g}$, error parameter $c$.
\Repeat{ \blu{(For parallel implementation, run $\ce{c}$ times simultaneously and take the first success.)}}
    \State Draw $X,Z\sim Q$.
    \State Let $R=\exp(g(X)-g(Z))$.
    \State Draw $U\sim \mathsf{Uniform}([0,1])$.
\Until{$U\le \rc{c}R$}
\OUTPUT $X$.
\end{algorithmic}
\label{a:ars}
\end{algorithm}

\begin{lem}
\label{l:ars}
Let $\wh P$ be the distribution of the output of Algorithm~\ref{a:ars}. 
Then 
\begin{align*}
\dd{P}{Q}(X) &= \fc{\E[R|X]}{\E R} & \dd{\wh P}{Q} (X)&= \fc{\E[\min\{R,c\}|X]}{\E[\min\{R,c\}]}.
\end{align*}
The acceptance probability is $p_{\textup{accept}}= \rc{c}\E \min\{R,c\} = \rc{c}(\E R - \E[(R-c)\one_{R\ge c}])$ and the TV distance is bounded by
\[
\TV(\wh P, P) \le \fc{\E[(R-c)\one_{R\ge c}]}{\E R}.
\]
For $c\ge1$, the acceptance probability is at least $\rc{2c}$.
\end{lem}
\begin{proof}
    We calculate
    \begin{align*}
        \E[R|X] &= \E[e^{g(X)-g(Z)}|X] = e^{g(X)}\E[e^{-g(Z)}] \\
        \E[R] &= \E[\E[R|X]] = 
        \E[e^{g(X)}] \E [e^{-g(Z)}],
    \end{align*}
    so $\fc{\E[R|X]}{\E[R]}  = \fc{e^{g(X)}}{\E[e^{g(X)}]} = \dd PQ (X)$.
    Note $\dd{\wh P}{Q}(X)$ is the probability of acceptance given $X$ divided by the total probability of acceptance, which we calculate:
    \begin{align*}
        p_{\textup{accept}}(X):=\Pj\ba{U\le \rc cR|X} & = \E\ba{\min\bc{\fc Rc, 1}|X} = \rc c 
        \E[\min\{R, c\}|X]
        \\
        p_{\textup{accept}}=\Pj\ba{U\le \rc cR} &= \E\ba{\Pj\ba{U\le \rc cR|X}}
        = \rc c \E[\E[\min\{R, c\}|X]] = \rc c \E[\min\{R, c\}].
    \end{align*}
    Dividing gives $\dd{\wh P}{Q} = 
    \fc{\E[\min\{R,c\}|X]}{\E[\min\{R,c\}]}$.
    Then
    \begin{align*}
        \TV(\wh P, P)
        &\le \E_{X\sim Q} \max\bc{0, \dd PQ(X)- \dd{\wh P}{Q}(X)}\\
        &\le \E_{X\sim Q} \max\bc{0, 
        \fc{\E[R|X]}{\E R} - \fc{\E[\min\{R,c\}|X]}{cp_{\textup{accept}}}}
        \\
        &\le \E_{X\sim Q} \max\bc{0, 
        \fc{\E[R|X]}{\E R} - \fc{\E[\min\{R,c\}|X]}{\E R}}& \text{because }cp_{\textup{accept}}\le \E R\\
        & \le \E_{X\sim Q}\fc{\E[(R-c)\one_{R\ge c}|X]}{\E R} = \fc{\E[(R-c)\one_{R\ge c}]}{\E R} .&
    \end{align*}
    Finally, note that $\Pj(R\ge 1)\ge \rc2$ by symmetry, so $p_{\textrm{accept}}\ge \rc c \Pj(R\ge 1) \ge \rc{2c}$. 
\end{proof}

\subsection{Concentration}
\label{s:conc}

To obtain concentration of the ratio in~\cref{l:ars}, we need a version of the Hanson-Wright inequality. We first state the classical inequality.
\begin{thm}[{Hanson-Wright Inequality, \cite[Theorem 6.2.1]{vershynin2018high}}]
\label{t:hwi}
There is a constant $c$ such that the following holds.
Let $X=(X_1,\ldots, X_n)\in \R^n$ be a random vector with independent, mean-zero, $K$-sub-gaussian coordinates. Let $A\in \R^{n\times n}$ be a matrix. Then for every $t\ge 0$,
\[
\Pj\pa{|\an{X,AX} - \E \an{X,AX}|\ge t}\le 2\exp\ba{-c \min\bc{\fc{t^2}{K^4\ve{A}_F^2}, \fc{t}{K^2\ve{A}}}}.
\]
\end{thm} 
We will use the following version, which is a consequence of~\cite[Corollary 2]{sambale2019modified} (by taking $\Ga(f) = \ve{\nb f}$ and noting that product distributions on $\Bn$ satisfy a uniform modified log-Sobolev inequality) and allows a general function $f$. 
\begin{thm}[\cite{sambale2019modified}]\label{t:hwi-general}
    There is a constant $c$ such that the following holds. Let $X=(X_1,\ldots, X_n)\in \Bn$ be a random vector with independent coordinates. Let $f:\Bn\to \R$ be a function. Then for every $t\ge 0$,
    \[
\Pj\pa{
    |f(X) - \E f(X)|\ge t
} \le 
2\exp\ba{
    -c \min \bc{
        \fc{t^2}{\E[\ve{\nb f}^2]},
        \fc{t}{\max_{x\in \Bn} \ve{\nb^2f}_F}
    }
}.
    \]
\end{thm}
Note the extra term $\E[\ve{\nb f}^2]$ compared to 
Theorem~\ref{t:hwi}
which requires the random variables to be centered.
This means that we cannot simply take $Q = \mu_{h} := \mu_{O,h}$ and $P = \mu_{J,h}$ for the reason that $\E[\ve{\nb_x(\an{x,Ax})}^2]$ can be $\Om(n)$, while we need concentration to $O(1)$. 
Instead, in order to apply Theorem~\ref{t:hwi-general} for $f=\ln \dd PQ$, we would like to $\nb f$ to be centered, that is, $\E_Q \nb f = 0$. 
For this, we need to solve the variational problem
    \begin{align}
        \label{e:fp}
        \E_{\mu_{h+h^*}} A^{\invdiameter}x &= h^* 
        .
    \end{align}
We do this by fixed point iteration. Note this is a special case of ``gradient" descent for Lipschitz and strongly monotone operators \cite{chen1997convergence,loizou2021stochastic}.
\begin{lem}[Fixed point iteration]
\label{l:fp}
Let $(X,d)$ be a metric space, $c>0$, and suppose $F:X\to X$ is $(1-c)$-Lipschitz (so it is a contraction mapping). 
Let $x_0\in X$ and $x_t = F^{(t)}(x_0)$. Then 
\[
d(F(x_t) , x_t) \le (1-c)^t d(F(x_0), x_0)
\]
and hence for $t=\Om\pa{\rc c \ln \pf{d(F(x_0),x_0)}{\ep}}$, we have 
$d(F(x_t) , x_t)\le \ep$. 
\end{lem}
\begin{proof}
    We have $d(F(x_t),x_t) = d(F^{(t+1)}(x_0), F^{(t)}(x_0)) \le (1-c)^t d(F(x_0), x_0)$ by induction. Hence, it suffices to choose $t$ such that $t\ln \prc{1-c} \ge \ln \pf{d(F(x_0),x_0)}{\ep}$, which gives the result.
\end{proof}
\begin{lem}\label{l:varl}
Suppose that $A$ is symmetric positive semi-definite with $A \preceq (1-c)I$. 
Let \[
F(u) = A^{\invdiameter} \tanh(h+u).
    \]
Then for $t = \Om\pa{\rc{c}\log\pf{\sqrt n}{\ep}}$, we have that $\wh h:=F^{(t)}(\mathbf 0)$ satisfies
    \begin{align}
        \label{e:afp}
        \ve{\E_{\mu_{h+\wh h}} A^{\invdiameter}x - \wh h} &\le \ep.
    \end{align}
\end{lem}
\begin{proof}
    Note that all diagonal entries of $A$ are contained in $[0,1-c]$, so $-(1-c)I \preceq A^{\invdiameter}\preceq (1-c)I$.
    Combining this with the fact that $\tanh$ is 1-Lipschitz, we obtain that $F$ is $(1-c)$-Lipschitz. %
    Note that $\ve{F(\mathbf 0) - \mathbf 0}\le \sqrt{n}$. The result then follows from~\cref{l:fp} and the fact that $F(u) = \E_{\mu_{h+u}}A^{\invdiameter}x$.
\end{proof}
Using this, for small enough $\ve{A}_F$, we can obtain the exponential tails necessary to bound the TV-distance in Lemma~\ref{l:ars}. This fits in with the general fact that Ising models with small $\ve{J}_F$ are well-approximated by product distributions~\cite{jain2019mean}, giving approximation guarantees for variational methods in this regime.
We state the following lemma more generally with a higher-order term, so that we can also apply it for the $p$-spin model.
\begin{lem}\label{l:pacc}
There is a constant $c_3$ such that the following holds. Suppose $H(x)=\an{h,x} + \rc 2\an{x,Ax} + H_{\ge 3}(x)$ where $A$ is symmetric and $H_{\ge 3}(x) = \sum_{|I|\ge 3} a_Ix^I$ contains the terms of degree $\ge 3$. If $\de<c_3$,
\[\max_{x\in \Bn}
\ve{\nb^2 H(x)}_F
\le \de\quad \text{ and }\quad \max_{x\in \Bn} \ve{\nb H_{\ge 3}(x)}\le \de,\] then the output of $\mathsf{QuadraticApproxRejectionSampler}$ (Algorithm~\ref{a:qars}) is at most $\ep$ in TV distance from $\mu$, and the acceptance probability in the call to $\mathsf{ApproxRejectionSampler}$ is at least $\rc{2c} = \rc 2 \pf{\ep}{2}^{\fc{\de}{c_3-\de}}$.  
\end{lem}
In the special case that $H_{\ge 3}(x)=0$, the assumption simplifies to $\ve{A^{\invdiameter}}_F\le \de$. 
\begin{proof}
    Let
    \[
f(x) = H(x) - \an{h+\wh h, x} = \rc 2 \an{x,Ax} - \an{\wh h, x} + H_{\ge 3}(x).
    \]
    To use Theorem~\ref{t:hwi-general}, we calculate $\E[\ve{\nb f}^2]$. 
    First note that 
    $\E_{x\sim \mathsf{Uniform}(\Bn)} \nb^2 H_{\ge 3}(x)=O$ (because for any $i,j\in [n]$ and $|I|\ge 3$, we have $\E_{x\sim \mathsf{Uniform}(\Bn)} x^{I\bs \{i,j\}} = 0$) so 
    by Jensen's inequality
    \begin{align*}
        \ve{A^{\invdiameter}}_F &= 
        \ve{A^{\invdiameter} + \E_{x\sim \mathsf{Uniform}(\Bn)} \nb^2 H_{\ge 3}(x)}_F\\
        &\le 
        \E_{x\sim \mathsf{Uniform}(\Bn)} \ve{A^{\invdiameter} +  \nb^2 H_{\ge 3}(x)}_F
        \le \max_{x\in \Bn} \ve{\nb^2 H(x)}_F \le \de.
    \end{align*}
    Let $\ol x = \E_{\mu_{h+\wh h}}x$. 
    From \cref{l:varl} we have that the output $\wh h$ of fixed point iteration satisfies $\ve{A^{\invdiameter}\ol x - \wh h}\le \de$.
    We then have
\begin{align}
\nonumber
    \E_{\mu_{h+\wh h}}[\ve{\nb f}^2]
    &= \E_{\mu_{h+\wh h}}\ba{\ve{A^{\invdiameter}x-\wh h + \nb H_{\ge 3}(x)}^2}\\
    \nonumber 
    &\le 2\E_{\mu_{h+\wh h}}\ba{\ve{A^{\invdiameter}(x-\ol x) + A^{\invdiameter}\ol x - \wh h}^2} + 2\E_{\mu_{h+\wh h}}\ba{\ve{\nb H_{\ge 3}(x)}^2}\\
    \nonumber
    &= 2\E_{\mu_{h+\wh h}}\ba{\ve{A^{\invdiameter}\ol x - \wh h}^2 + \ve{A^{\invdiameter}(x-\ol x)}^2 + \ve{\nb H_{\ge 3}(x)}^2}\\
    &\le 2\ba{\ve{A^{\invdiameter}\ol x-\wh h}^2 + \ve{A^{\invdiameter}}_F^2 + \E_{\mu_{h+\wh h}} \ve{\nb H_{\ge 3}(x)}^2}\le 6\de^2
    \label{e:E-nbf2}
\end{align}
using the fact that the entries of $x-\ol x$ are independent, mean 0, with variance at most 1.
Hence, by Theorem~\ref{t:hwi-general}, $f(X) - \E f(X)$ is $O(\de)$-sub-exponential, and so is $f(Z)-f(X)$, and there exists $c_3$ so that
\[
\Pj\pa{|f(Z)-f(X)|\ge t} \le 2e^{-\fc{c_3t}{\de}}.
\]
Then for $c=\pf{2}{\ep}^{\fc{\de}{c_3-\de}}$, 
\begin{align}
\nonumber
    \E[(R-c)\one_{R\ge c}]
    &\le \int_{\ln c}^\iy e^{t}\cdot  \Pj(f(Z)-f(X)\ge t)\,dt\\
    &\le \int_{\ln c}^\iy e^t 2 e^{-\fc{c_3t}{\de}}\,dt
    \le 2\int_{\ln c}^\iy e^{-(\fc{c_3}{\de} - 1)t} \,dt = 2c^{-\pa{\fc{c_3}{\de}-1}} = \ep. 
    \label{e:power-tail}
\end{align}
Moreover, by Jensen's inequality, $\E R\ge e^{\E[f(Z)-f(X)]}=1$.  Hence by \cref{t:hwi-general}, the output is at most $\ep$ in TV distance from $\mu$ and the acceptance probability is at least $\rc{2c}$.
\end{proof}

\subsection{Analysis of the Parallel Ising Sampler}
\label{s:proof}

 \begin{lem}\label{l:prod}
 Let $S\subeq [n]$, fix $x_{S^c}\in \{\pm 1\}^{S^c}$, and let $P$ be the distribution on $\{\pm 1\}^S$ with mass function 
 $p(x)=\mu_{J,h}(X_S=x |X_{S^c}=x_{S^c})$, and let $Q$ be the product distribution in on $\{\pm 1\}^S$ with mass function $q(x) \propto \exp\pa{\an{J_{S\times S^c}x_{S^c}+ h_S, x}}$.
 Then the following hold.
 \begin{enumerate}
     \item $\dd PQ(x) \propto \exp\pa{\rc 2 \an{x, J_{S\times S}x}}$. 
     \item $\KL(P\|Q)\le \ve{J_{S\times S}}\cdot |S|$.
 \end{enumerate}
\end{lem}
 \begin{proof}
Because $x_{S^c}$ is constant, expanding the quadratic gives
\[
\mu_{J,h}(X_S=x_S|X_{S^c}=x_{S^c}) \propto 
\exp\pa{\rc 2 (2\an{x_S,J_{S\times S^c}x_{S^c}} + \an{x_S, J_{S\times S}x_S}) + \an{h_S,x}}.
\]
Dividing by $q(x_S)$ gives (1). 

For (2), we note that for $x\in \{\pm 1\}^S$,  $\ab{\rc 2\an{x, J_{S\times S}x}}\le \rc 2 \ve{J_{S\times S}} \ve{x}^2\le \rc 2  \ve{J_{S\times S}} |S|$.
Hence 
\[
\dd PQ(x) = \fc{\exp\pa{\rc 2 \an{x,J_{S\times S}x}}}{\int \exp\pa{\rc 2 \an{x,J_{S\times S}x}}\,dQ(x)} \le 
\fc{e^{\rc 2\ve{J_{S\times S}}|S|}}{e^{-\rc 2\ve{J_{S\times S}}|S|}} = 
e^{\ve{J_{S\times S}}|S|}
\]
and $\KL(P\|Q) = \E_P \ln \dd PQ\le \ve{J_{S\times S}}|S|$.
 \end{proof}

\begin{lem}[{Bernstein's inequality for supermartingales \cite[(1.6)]{freedman1975tail}}]
\label{l:bern}
    Let $X_n$ be a martingale adapted to $\cal F_n$. Suppose that $|X_{n+1}-X_{n}|\le L$ and $\E [|X_{n+1}-X_{n}|^2 |\cal F_n]\le \si^2$ with probability 1. Then 
    \[
\Pj(X_n-X_0 \ge t) \le \exp\pa{-\fc{t^2}{2(tL + n\si^2)}}. 
    \]
\end{lem}

We are now ready to prove our main theorem on the parallel Ising sampler. 

\begin{proof}[Proof of Theorem~\ref{t:main}.]
We first note that all lines in Algorithm \ref{a:pising} take logarithmic time with $\poly(n)$ processors (e.g., by a parallel implementation of matrix-vector multiplication). Note that a random subset of specified size $s$ can be selected by generating a random number for each index, using a parallel sorting algorithm~\cite{bitton1984taxonomy}, and then selecting the smallest $s$ elements. 
We will ignore logarithmic overhead for the rest of the proof. 


~

\noindent \textbf{Running time is bounded with high probability.} 
We consider a tree associated with a run of the algorithm, where each node is labeled with a set, constructed as follows. Each node represents a time that $\mathsf{ParallelIsingSampler}$ is called, and each leaf node represents a time that $\mathsf{ApproxRejectionSampler}$. 
Start with a root node $v_1$ labeled with $S_1=[n]$. A node has $T$ children, where $T$ is the number calculated in line \ref{st:T} of the algorithm. Each node is labeled with subset of indices marking out the submatrix $J_{S\times S}$ it is given.

Now consider exploring the tree in the following breadth-first manner. 
We will define a list $B_t$ which will contain the vertices at the boundary of explored territory and a filtration $\cal F_t$. 
Let $B_0=(v_1)$ and $\cal F_0$ be the trivial $\si$-algebra. Given $B_t$ and $\cal F_t$, if $B_t$ is non-empty, define $B_{t+1}$ and $\cal F_{t+1}$ as follows. Let $v_{t+1}$ be the first vertex in the list $B_t$, and let $B_{t+1}$ be defined from $B_t$ by removing $v_{t+1}$ from $B_t$ and adding its children. Let $S_{t+1}$ denote the set of indices associated with $v_{t+1}$, considered as a set-valued random variable, and $\cal F_{t+1}=\si(\cal F_t, S_{t+1})$. Let $M_t=|B_t|$. We have $M_1=\fl{ C_2\ln \pf{n}{\ep}\fc{n}s}$, and wish to bound the first time $\tau$ such that $M_\tau=0$. We redefine $M_{\tau+k} =-k$ (for sake of making $M_t$ a supermartingale, as we will show below).

For $t\ge 2$, consider $M_{t}-M_{t-1}|\cal F_{t-1}$. Let $v$ denote the parent of $v_{t}$, and suppose $v$ is associated with the set $R$, with $|R|=m$. 
Then $|S_t|=s:=\ce{\fc{c_1 m}{\pfrobslack
\ln \pf{n}{\ep}\ve{J_{R\times R}^{\invdiameter}}_F}}$. 
(We choose $c_1\le \rc 2c_3$ to ensure that we always have $s\le m$.)
Let $D_t$ be the number of new children added. 
If $\ve{J_{S_t\times S_t}^{\invdiameter}}_F\le c_3$ or $s=1$, then $v_t$ is a leaf and $D_t=0$. 
Now consider $\ve{J_{S_t\times S_t}^{\invdiameter}}_F> c_3$. 
In the current call to the algorithm, $s'=\ce{\fc{c_1s}{\pfrobslack\ln \pf n\ep \ve{J_{S_t\times S_t}^{\invdiameter}}_F}}$. 
Then $D_t\le
C_2\ln \pf{n}{\ep} \fc{s}{s'}\le 
{ \fc{C_2 \pfrobslack\ln \pf{n}{\ep}^2\ve{J_{S_t\times S_t}^{\invdiameter}}_F}{c_1}}$. In either case, $M_t-M_{t-1}=D_t-1$. 
We have that
\begin{align}
\label{e:dt-ineq}
    D_t &\le \fc{C_2 \pfrobslack\ln \pf{n}{\ep}^2
    \ve{J_{S_t\times S_t}^{\invdiameter}}_F \one\ba{\ve{J_{S_t\times S_t}^{\invdiameter}}_F> c_3}}{c_1}.
\end{align}
Hence, by Cauchy-Schwarz and Chebyshev's inequality,
\begin{align}
\nonumber
    \E[D_t|\cal F_{t-1}] 
    &\le  \fc{C_2\pfrobslack\ln \pf{n}{\ep}^2}{c_1}
    \E\ba{ \ve{J_{S_t\times S_t}^{\invdiameter}}_F^2 | \cal F_{t-1}}^{1/2} \Pj\ba{\ve{J_{S_t\times S_t}^{\invdiameter}}_F> \fc{c_3}{\frobslack}\Big| \cal F_{t-1}}^{1/2}\\
    \nonumber
    &\le \fc{C_2\pfrobslack\ln \pf{n}{\ep}^2}{c_1}
    \E\ba{ \ve{J_{S_t\times S_t}^{\invdiameter}}_F^2| \cal F_{t-1}}^{1/2} 
    \cdot 
    \fc{\E\ba{\ve{J_{S_t\times S_t}^{\invdiameter}}_F^2| \cal F_{t-1}}^{1/2} }{c_3/\pa{\frobslack}}\\
    &= \fc{C_2\pfrobslack^2\ln \pf{n}{\ep}^2}{c_1c_3}
    \E\ba{ \ve{J_{S_t\times S_t}^{\invdiameter}}_F^2| \cal F_{t-1}}.
    \label{e:ET}
\end{align}
Now because $S_t$ is uniformly chosen at random from subsets of $R$ of size $s$,
\begin{align}
\nonumber
    \E\ba{ \ve{J_{S_t\times S_t}^{\invdiameter}}_F^2| \cal F_{t-1}}
    &=\E_{S\sim \mathsf{Uniform}\binom{R}{s}} \ba{\ve{J_{S\times S}^{\invdiameter}}_F^2}\\
&= \sum_{i,j\in R, i\ne j} \pf{s}{m}^2 J_{ij}^2 = \pf{s}{m}^2 \ve{J_{R\times R}^{\invdiameter}}_F^2 \le \fc{4c_1^2}{\pa{\frobslack}^2\ln\pf n\ep^2},
\label{e:JF}
\end{align}
where we use the fact that $J_{ii}^{\invdiameter}=0$, all off-diagonal entries have probability $\pf{s}{m}^2$ of being included in $S_t$, and $s>1$. 
Combining~\eqref{e:ET} and~\eqref{e:JF} gives
\begin{align*}
    \E[D_t|\cal F_{t-1}] \le \fc{4C_2c_1}{c_3} .
\end{align*}
Choosing $c_1$ small enough (depending on $C_2, c_3$), we can ensure that $\E[M_t-M_{t-1}|\cal F_{t-1}] = \E[D_t-1|\cal F_{t-1}]\le -\rc 2$, so that $M_t+\fc t2$ is a supermartingale for $t\ge 1$. By Doob's decomposition we can write $M_t = A_t + M_t'$ where $A_{t+1}\le A_1-\fc t2$ is a predictable decreasing sequence and $M_t'$ is a martingale.

We now bound the variance. Using~\eqref{e:dt-ineq},
\begin{align*}
\E[(M_{t}'-M_{t-1}')^2|\cal F_{t-1}]&\le 
    \E[D_t^2|\cal F_{t-1}]\\
    &\le \fc{C_2^2\pfrobslack^2\ln \pf n\ep^4 }{c_1^2} 
    \E\ba{ \ve{J_{S_t\times S_t}^{\invdiameter}}_F^2| \cal F_{t-1}}
    \\
    &\le 4C_2^2 \ln \pf{n}{\ep}^2,
\end{align*}
where we use the bound~\eqref{e:JF}.
Finally,  $|M_{t+1}'-M_t'|\le \fc{C_2 \pfrobslack\ln \pf n\ep^2}{c_1}\ve{J^{\invdiameter}}_F$ with probability 1. Let $T_0$ be the $T$ computed in line~\ref{st:T} in the first step of the algorithm. 
By Bernstein's inequality for martingales (\cref{l:bern}), for $t\ge C\ln^4\pf n{\ep} \max\bc{\ve{J^{\invdiameter}}_F,1}\ge T_0$ for an appropriate constant $C$ (depending on $C_2,c_1, C_4$), 
\begin{align*}
    \Pj(M_{t+1}> 0) 
    &=
    \Pj\pa{(M_{t+1}-M_1)> -T_0}
    \le 
    \Pj\pa{M_{t+1}'-M_1'>\fc{t}{2}-T_0}
    \le 
    \eph. 
\end{align*}
This shows that with probability $\ge 1-\fc{\ep}4$, there are at most $t_{\max} = C\ln^4\pf n{\ep} \max\{\ve{J}_F,1\}$ nodes.

Finally, we note that in the call to $\mathsf{ApproxRejectionSampler}$, the parameter needed to obtain error $\ep_{\mathrm{step}}$ is
\begin{align*}
    c &= \pf{2}{\ep_{\mathrm{step}}}^{\fc{\de}{c_3-\de}} = \pf{2}{\ep_{\mathrm{step}}}^{\rc{\log(2/\ep_{\mathrm{step}})}} = e
\end{align*}
and the acceptance probability is $\ge \rc{2c}= \rc{2e}$.

The number of tries until acceptance is a geometric random variable, which is subexponential, so standard concentration bounds show that the total number of tries is at most $O(\ln\prc{\ep})$ times the number of calls, with probability $\ge 1-\fc{\ep}2$. 
Putting everything together, we obtain $O(\max\{\ve{J}_F,1\}\poly\log\pf{n}{\ep})$ running time with probability $\ge 1-\ep$.

~

\noindent \textbf{Output is close in TV distance.} Let $A$ be a large constant to be determined.

Now consider coupling $y=y^{(0)}$ with a sequence of random variables $y^{(1)},\ldots$, defined inductively as follows. Start with all vertices of the tree of recursive calls unmarked. Now given $y^{(t)}$, choose a node (in a fixed manner) all of whose children are marked, and mark it; then replace the output of that call to $\mathsf{ParallelIsingSampler}$ by a sample from the true distribution. 
We now choose constants so that $\TV(\cal D(y^{(t)}), \cal D(y^{(t+1)}))\le \fc{\ep}{n^A}$. There are two kinds of replacements to consider, a leaf node and a non-leaf node. 

A leaf node corresponds to a call to  $\mathsf{ApproxRejectionSampler}$. If $c_3$ is small enough and 
\blu{$C_4=A$}, then by \cref{l:pacc} and~\ref{l:ars}, the output of $\mathsf{ApproxRejectionSampler}$ is within $\fc{\ep}{n^A}$ of the $\mu_{J_{R\times R},h}$. 

A non-leaf node corresponds to $T$ recursive calls to $\mathsf{ParallelIsingSampler}$. 
Here we must appeal to mixing for the Ising model. By Theorem~\ref{t:atoe-ising}, approximate tensorization of entropy holds with constant $\rc{c}$. 
Hence by Theorem~\ref{t:k-glauber}, 
there is a constant $C_0'$ such that if $C_0 = C_0'c$, then for any $s$, $t\cdot \fc ns$ steps of $s$-Glauber dynamics results in a distribution $\nu_t$ satisfying
\begin{align*}
    \TV(\nu_t\|\mu_{J,h}) \le 
    \sqrt{\rc 2 \KL(\nu_t\|\mu_{J,h})} 
    \le \sqrt{\rc 2 \KL(\nu_0\|\mu_{J,h}) e^{-C_0t}}.
\end{align*}
With the product initialization, we have by Lemma~\ref{l:prod}(2) (applied to the whole matrix) that $\KL(\nu_0\|\mu_{J,h})\le \ve{J}n\le 2n$. Hence there exists a constant $C_2'$ such that if $C_2=C_2'A/c$, then with $T = C_2\ln \pf n\ep \fc{n}{s}$ steps, $\TV(\nu_T\|\mu_{J,h})\le 
\fc{\ep}{n^A}$. 
By \cref{l:prod}(1), for the Ising model $\mu_{J_{R\times R}, h}$ the conditional distribution of $X_S$ given $X_{R\bs S}=y_{R\bs S}$ is exactly the Ising model $\mu_{J_{S\times S}, J_{S\times R\bs S}y_{R\bs S}+h_S}$. Given that the conditional distributions are sampled exactly, then the only error is that from not having fully mixed, which we set to be $\fc{\ep}{n^A}$. 

This chain of coupled random variables establishes $\TV(\cal D(y), \cal D(y^{(t)}))\le \fc{t\ep}{n^A}$. Moreover, for $t>t_{\max}$, $\TV(\cal D(y^{(t)}),\mu_{J,h})\le \fc{\ep}{2}$ by our high-probability bound, as the root node in $y^{(t)}$ will have been replaced with a perfect sample with probability $\ge 1-\fc{\ep}2$. 
It remains to note that $t_{\max}=C\ln^4\pf n{\ep} \max\{\ve{J}_F,1\}$ with $C$ depending polynomially on $A$. Hence we can choose $A$ such that $\fc{t_{\max}\ep}{n^A}\le \eph$, and this finishes the proof. 
\end{proof}

\section{Parallel sampling for $p$-spin model}

    

Our starting point is the following theorem, which bootstraps the spectral gap of~\cite{adhikari2022spectral} into approximation tensorization of entropy.
\begin{thm}[\cite{anari2023universality}]
\label{t:atoe-p-spin}
    There is an absolute constant $A$ such that if 
    \begin{align}
    \label{e:Cbe}
    C(\be):&=\sum_{p=2}^{\iy} \sqrt{p^3\ln p}\cdot \be_p\le A
\end{align}
    and \begin{align*}D(\be) := \sum_{p=2}^{\iy} \sqrt{2^p p^3 \ln p}\cdot \be_p<\iy,
    \end{align*}
    then with probability $\ge 1-\exp(-\Om(n))$ over $g$, $\mu_{\be, g,h}$ satisfies approximate tensorization of entropy with constant depending only on $D(\be)$.
\end{thm}

\subsection{Concentration}
\label{s:conc-p-spin}
As in Section~\ref{s:conc}
We would like to bound for all $|S|=s$ and $x\in \Bn$ 
the vector 
$\nb H_{x_{S^c}} (x_S)$ and the matrix 
$\nb^2 H_{x_{S^c}} (x_S)$ in Frobenius norm. The arguments are similar to those in~\cite{tomioka2014spectral,adhikari2022spectral,anari2023universality}. 
We write
\begin{align*}
    (\nb H_{x_{S^c}}(x_S))_i
    &= h_i + \sum_{p=2}^{\iy} \sumr{J\ni i}{|J|=p} g_J x^{J\bs \{i\}}\\
    &= h_i + \ub{\sum_{p=2}^{\iy} \fc{\be_p \sqrt{p!}}{n^{\fc{p-1}2}} \sumr{J\ni i}{|J|=p, |J\cap S|\le 2} g_J x^{J\bs \{i\}}}{=:(\ga_{S,\le 2}(x))_i} + \ub{\sum_{p=3}^{\iy} \fc{\be_p \sqrt{p!}}{n^{\fc{p-1}2}} \sumr{J\ni i}{|J|=p, |J\cap S|\ge 3} g_J x^{J\bs \{i\}}}{=:(\ga_{S,\ge 3}(x))_i}\\
    (\nb^2 H_{x_{S^c}}(x_S))_{ij}
    &= \sum_{p=2}^{\iy} \fc{\be_p \sqrt{p!}}{n^{\fc{p-1}2}} \sumr{J\supeq \{i,j\}}{|J|=p} g_J x^{J\bs \{i,j\}}.
\end{align*}
Note that $\ga_{\ge3}(x)$ is the higher-degree error term that now arises when using Lemma~\ref{l:pacc}. 
A complication is that for a given $x\in \Bn$, the entries of $\nb H_{x_{S^c}}(x_S)$ and $\nb^2 H_{x_{S^c}}(x_S)$ are not independent, because $g_J$ appears in all entries where $i\in J$ or $\{i,j\}\subeq J$. 
To bound $\ga_{\ge 3}(x)$, write it as 
\begin{align*}
    \ga_{\ge 3}(x) &= B_{S,x}^{(1)} g &\text{where } B_{S,x}^{(1)} &\in \R^{S\times \binom{[n]}{\ge 2}}\\
    && (B_{S,x}^{(1)})_{i,J}& = \fc{\be_p\sqrt{p!}}{n^{\fc{p-1}2}} \one_{i\in J, |J\cap S|\ge 3} x^J.
\end{align*}
Here, $\binom{[n]}{\ge p}$ denotes the subsets of size at least $p$, and 
we view $g$ as a random Gaussian vector in $\R^{\binom{[n]}{\ge 2}}$.
To bound $\ve{\nb^2 H_{x_{S^c}}(x_S)}_F$, we consider it as a vector in $\R^{\binom{S}2}$ and write 
\begin{align*}
    \vc(\nb^2 H_{x_{S^c}}(x_S))
    &= B_{S,x}^{(2)}g
    & \text{where }
    B_{S,x}^{(2)} &\in \R^{\binom S2\times \binom{[n]}{\ge 2}}\\
    &&(B_{S,x}^{(2)})_{I,J} &=\fc{\be_p\sqrt{p!}}{n^{\fc{p-1}2}} \one_{I\in J} x^J.
\end{align*}
For a matrix $B$ let $|B|$ denote the matrix whose entries are the absolute values of entries of $B$. Let $A_{S}^{(i)} =  |B_{S,x}^{(i)}|$ for $i=1,2$. (Note this does not depend on $x$.) For the bounds below, let
\[
C_k(\be) = \sum_{p\ge 2} p^k \be_p^2.
\]
We note that $C_3(\be)=O(C(\be)^2)$ where $C(\be)$ is defined in~\eqref{e:Cbe}.
\begin{lem}
\label{l:B1}
    For all $|S|=s$, $x\in \Bn$, the following bounds hold for $B=B^{(1)}_{S,x}$:
    \begin{align*}
        \ve{B}_F^2 &\lesssim \fc{s^3C_3(\be)}{n^2} \\
        \ve{B}^2 = \ve{B^\top B} = \ve{BB^\top} &\lesssim \fc{s^2C_3(\be)}{n^2}
        \\
        \ve{B^\top B}_F &\lesssim \fc{s^{2.5}C_3(\be)}{n^2}
    \end{align*}
\end{lem}
\begin{proof}
    By symmetry, each row of $B$ has the same 2-norm. 
    Let $i\in S$. 
    Counting the number of sets $J$ such that $i\in J$, $|J|=p$, and $|J\cap S|=p_1$ gives (using Vandermonde convolution to evaluate the sum)
    \begin{align*}
        \ve{B}_F^2 &= s \sum_{p\ge 3} \fc{\be_p^2 p!}{n^{p-1}} \sumr{p=p_1+p_2}{p_1\ge 3} 
        \binom{s-1}{p_1-1} \binom{n-s}{p_2}\\
        &\le s \sum_{p\ge 3} \fc{\be_p^2 p!}{n^{p-1}}\sumr{p=p_1+p_2}{p_1\ge 3}
        \fc{(s-1)(s-2)}{(p_1-1)(p_1-2)} \binom{s-3}{p_1-3} \binom{n-s}{p_2}\\
        &\le \fc{s^3}{2} \sum_{p\ge 3} \fc{\be_p^2 p!}{n^{p-1}} \binom{n-3}{p-3} \le \fc{s^3 C_3(\be)}{2n^2}.
    \end{align*}
    For $\ve{B}$, note $BB^\top$ is symmetric so its operator norm can be bounded by the $\ell_\iy\to \ell_\iy$ norm: $\ve{B}^2 = \ve{BB^\top}
        \le \max_i \sum_j |(BB^\top)_{i,j}|$. Now
    \begin{align*}
        \max_i \sum_j |(BB^\top)_{i,j}|
        &\le  \sum_{p\ge 3}
        \fc{\be_{p}^2 p!}{n^{p-1}}
        \sumr{p_1+p_2=p}{p_1\ge 3} 
        \ab{\set{(K,j)}{K\ni i, |K\cap S|=p_1, |K\cap S^c|=p_2, j\in K}}\\
        &= \sum_{p\ge 3} \fc{\be_p^2 p!}{n^{p-1}}\sumr{p_1+p_2=p}{p_1\ge 3} \binom{s-1}{p_1-1}\binom{n-s}{p_2}p_1\\
        &= \sum_{p\ge 3} \fc{\be_p^2 p!}{n^{p-1}}\sumr{p_1+p_2=p}{p_1\ge 3}\fc{s-1}{p_1-1} \binom{s-2}{p_1-2}\binom{n-s}{p_2}p_1\\
        &\le \fc{3s}2 \sum_{p\ge 3} \Bigg( \fc{\be_p^2 p!}{n^{p-1}}\sumr{p_1+p_2=p}{p_1\ge 2} \binom{s-2}{p_1-2}\binom{n-s}{p_2} - \binom{n-s}{p-2}\Bigg)\\
        &\le \fc{3s}{2} \pa{\sum_{p\ge 3}\fc{\be_p^2 p!}{n^{p-1}}\pa{\binom{n-2}{p-2} - \binom{n-s}{p-2}}}\\
        &\le \fc{3s}2 \pa{\sum_{p\ge 3} \fc{\be_p^2 p!}{n^{p-1}} \binom{n-3}{p-3} (s-2)}\le \fc{3s^2C_3(\be)}{2n^2}.
    \end{align*}
    The last inequality follows from $\ve{B^\top B}_F\le \ve{B}\ve{B}_F$.
\end{proof}
For any $x\in \{\pm 1\}^n$ and $s\le \sqrt n$, we hence have by Hanson-Wright~\ref{t:hwi} that, for some constant $C'(\be)$ such that $C'(\be)\to 0$ as $C(\be)\to 0$,  
    \begin{align*}
    \Pj\ba{|\ga_{S,\ge 3}(x)|^2 - 
    \E |\ga_{S,\ge3}(x)|^2}
    &= 
        \Pj[|g^\top B_{S,x}^{(1)\top} B_{S,x}^{(1)} g-\E g^\top B_{S,x}^{(1)\top } B_{S,x}^{(1)} g|\ge t]\\
        &\le 
        2\exp\ba{-c \min\bc{\fc{t^2}{\ve{B_{S,x}^{(1)\top } B_{S,x}^{(1)}}_F^2}, \fc{t}{\ve{B_{S,x}^{(1)\top } B_{S,x}^{(1)}}}}}
        \\
        &\le 2\exp\pa{-\fc{3t}{C'(\be)}},
    \end{align*}
    by substituting the bounds in Lemma~\ref{l:B1}.
    Hence, for $s\le \sqrt n$ and $t\ge C'(\be)$, this is $\le 8^{-n}$.
    Moreover,
    \begin{align*}
        \E g^\top B_{S,x}^{(1)\top } B_{S,x}^{(1)} g &= 
        \Tr(B_{S,x}^{(1)\top } B_{S,x}^{(1)}) = \ve{B_{S,x}^{(1)}}_F^2 \lesssim   \fc{s^3 C'(\be)^2}{n^2} 
        \lesssim C'(\be)^2.
    \end{align*}
    Taking a union bound over all $|S|=s$ and $x\in \{\pm 1\}^n$ shows that for some constant $C''(\be)\to 0$ as $C(\be)\to 0$ that 
    \begin{align}
    \label{e:ga3}
        \Pj\pa{
\forall S\sub [n], |S| = s, \,
\forall x\in \{\pm 1\}^{n}, \,
\ve{\ga_{S,\ge 3}(x)}^2\le 
C''(\be)
        }\ge 1-2^{-n}.
    \end{align}
\begin{lem}
For all $|S|=s$ and $x\in \Bn$, the following bounds hold for $B=B_{S,x}^{(2)}$:
\begin{align*}
    \ve{B}_F^2 &
    \lesssim \fc{s^2C_2(\be)}{n}\\
    \ve{B}^2 = 
    \ve{B^\top B} = 
    \ve{BB^\top} &\lesssim
     \fc{C_3(\be)}{n}\\
    \ve{B^\top B}_F &\lesssim 
     \fc{sC_3(\be)}{n}.
\end{align*}
\end{lem}
\begin{proof}
    By symmetry, each row of $B$ has the same 2-norm. 
    Let $I\in \binom{S}{2}$. 
    The number of sets of size $p$ containing $I$ is $\binom{n-2}{p-2}$, so 
    \[
\ve{A_I}^2 = \sum_{p\ge 2} \binom{n-2}{p-2}\fc{\be_p^2p!}{n^{p-1}}
\le \sum_{p\ge 2} \fc{n^{p-2}}{(p-2)!} \fc{\be_p^2 p!}{n^{p-1}} \le \fc{C_2(\be)}{n}.
    \]
    Multiplying by $\binom s2\le \fc{s^2}2$ gives $\ve{B}_F^2$. 

    For $\ve{B}$, we bound by the $\ell_\iy\to \ell_\iy$ norm: $\ve{B}^2 = \ve{BB^\top}\le  \max_I \sum_J |(BB^\top )_{I,J}|$. Now
    \begin{align*}
        \sum_J |(BB^\top)_{I,J}|
        &\le \sum_{(K,J): K\supeq I, |J|=2, J\subeq K}
        \fc{\be_{|K|}^2 |K|!}{n^{|K|-1}}\\
        &=\sum_{p\ge 2}
        \fc{\be_{p}^2 p!}{n^{p-1}}
        \sum_{p_1+p_2=p} 
        \ab{\set{(K,J)}{K\supeq I, |K\cap S|=p_1, |K\cap (S^c)|=p_2, |J|=2, J\subeq K}}\\
        &=\sum_{p\ge 2}
        \fc{\be_{p}^2 p!}{n^{p-1}}
        \sum_{p_1+p_2=p} 
        \binom{s}{p_1-2} \binom{n-s}{p_2} \binom{p_1}2\\
        &= \sum_{p\ge 2}
        \fc{\be_{p}^2 p!}{n^{p-1}}
\ba{\binom{n-s}{p-2} + s\binom{n-s}{p-3}3 +     
        \sum_{p_1+p_2=p, p_1\ge 4} 
        \binom{s}{p_1-2} \binom{n-2}{p_2} \binom{p_1}2}\\
                &\le \sum_{p\ge 2}
        \fc{\be_{p}^2 p!}{n^{p-1}}
\ba{\binom{n-s}{p-2} + s\binom{n-s}{p-3}3 +     
        \sum_{p_1+p_2=p, p_1\ge 4} 
        \binom{s}{p_1-4} \binom{n-s}{p_2} }\\
        &\le \sum_{p\ge 2}
        \fc{\be_{p}^2 p!}{n^{p-1}}
\ba{\binom{n-s}{p-2} + s\binom{n-s}{p-3}3 +     
        \binom{n}{p-4}}\\
        &\le
        \sum_{p\ge 2} \fc{\be_{p}^2 p!}{n^{p-1}} \fc{n^{p-2}}{(p-2)!}
        +
        3\sum_{p\ge 3} \fc{s\be_{p}^2 p!}{n^{p-1}} \fc{n^{p-3}}{(p-3)!}
        + 
        \sum_{p\ge 4} \fc{\be_{p}^2 p!}{n^{p-1}} \fc{n^{p-4}}{(p-4)!}\\
        &\le 
        2\cdot \sum_{p\ge 2} \fc{\be_{p}^2 p(p-1)}{n}
        + 
        3\cdot \sum_{p\ge 3} \fc{s\be_{p}^2 p(p-1)(p-2)}{n^2}
        \le \fc{2C_2(\be) + 3C_3(\be)}{n}\lesssim \fc{C_3(\be)}{n}.
    \end{align*}

    Finally, the last inequality follows from $\ve{B^\top B}_F \le \ve{B} \ve{B}_F $. 
\end{proof}
The same argument using Hanson-Wright and a union bound then shows
    \begin{align}
\label{e:HH}
        \Pj\pa{
\forall S\sub [n], |S| = s, \,
\forall x\in \{\pm 1\}^{n}, \,
\ve{\nb^2 H_{x_{S^c}}(x_S)}_F^2\le 
C''(\be)
        }\ge 1-2^{-n}
    \end{align}
    for $C''(\be)\to 0$ as $C(\be)\to 0$.
    \begin{lem}
        Consider the call to $\mathsf{QuadraticApproxRejectionSampler}$ in line~\ref{st:call-qars-pspin} of Algorithm~\ref{a:ppspin}. There is $\de$ such that if $C(\be)<\de$, with probability $1-2\cdot 2^{-n}$ over $g$, no matter the choice of $|S|\le \sqrt n$, the output of $\mathsf{QuadraticApproxRejectionSampler}$ is at most $\ep_{\textrm{step}}$ in TV distance from $\mu_{\be, g, h}(X_{S}=\cdot | X_{S^c} = x_{S^c})$ and the acceptance probability in $\mathsf{ApproxRejectionSampler}$ is at least $\fc{\estep}{4}$. 
    \end{lem}
    \begin{proof}
        From~\eqref{e:ga3} and~\eqref{e:HH}, we get 
        \begin{align}
        \label{e:good}
\Pj\pa{\forall S\sub [n], |S| = s, \,
\forall x\in \{\pm 1\}^{n}, \,
\ve{\ga_{S, \ge 3}(x)}^2\le C''(\be), \, 
\ve{\nb^2 H_{x_{S^c}}(x_S)}_F^2\le 
C''(\be) }\ge 1-2\cdot 2^{-n}.
        \end{align}
        Note that choosing $\de' = \fc{c_3}2$ and ensuring 
        $C''(\be)\le \de^{\prime 2}$, we have that the accuracy parameter passed to $\mathsf{ApproxRejectionSampler}$ is 
        $c=\pf{2}{\estep}^{\fc{\de'}{c_3-\de'}} = \fc{2}{\estep}$. 
        Moreover, there exists $\de>0$ such that $C(\be)<\de$ implies $C''(\be)<\de^{\prime2}$.
        For $C''(\be) \le \de' = \fc{c_3}2$, we hence have by \cref{l:pacc} that, under the good event in~\eqref{e:good} that for any $S$ and $x_{S^c}$, the output of $\qars$ is at most $\estep$ in TV distance from $\mu_{\be, g, h}(X_{S}=\cdot | X_{S^c} = x_{S^c})$. The acceptance probability is at least $\fc{1}{2c} = \fc{\estep}4$.
    \end{proof}

\subsection{Analysis of parallel $p$-spin sampler}

We will need to bound the KL divergence starting from the initial distribution (cf. Lemma~\ref{l:prod}). 
\begin{lem}\label{l:p-spin-init}
    Given $c_1$, there is a constant $C_2$ such that 
    \[
\Pj\pa{
    \KL(\mu_{\be, g, h}\| \mu_{h}) \le n\cdot C_2 \sqrt{C_0(\be)}
} \ge 1-e^{-c_1n}.
    \]
    where $\mu_h(x) \propto e^{\an{h,x}}$ is a product distribution.
\end{lem}
We bound this through a standard argument bounding the maximum of the Hamiltonian.
\begin{lem}
\label{l:max-H}
Given $c_1>0$, there is a constant $C_2>0$ such that 
    \[
\Pj\pa{\forall x\in \{\pm 1\}, \,|H_{\be, g, h}^{\ge2}(x)| \le n\cdot  C_2 \sqrt{C_0(\be)} } \ge 1-e^{-c_1 n}.
    \]
\end{lem}
\begin{proof}
For a fixed $x\in \Bn$, note $H^{\ge 2}(x)$ is Gaussian. We compute the variance of $H^{\ge 2}(x)$ over the randomness in $g$:
        \begin{align*}
            \Var\pa{H_{\be, g, h}^{\ge 2}(x)}
            & = \sum_{p\ge 2} \fc{\be_p^2 p!}{n^{p-1}} \binom{n}{p} \le C_0(\be)n.
        \end{align*}
        A standard tail bound gives that for $t\ge 1$, $\Pj\pa{|H_{\be, g, h}^{\ge 2}(x)|\ge t\sqrt{C_0(\be) n}}\le 2\cdot \rc{\sqrt{2\pi}}e^{-\fc{t^2}{2}}$ and a union bound gives 
        \[
\Pj\pa{\forall x\in \{\pm 1\}, \,|H_{\be, g, h}^{\ge2}(x)| \le t \sqrt{C_0(\be) n}} \ge 1-2^n e^{-\fc{t^2}{2}}.
        \]
        Choosing $t$ on the order of $\sqrt n$ then gives the result.
\end{proof}

\begin{proof}[Proof of~\cref{l:p-spin-init}]
    Note that under the event in~\cref{l:max-H},
\[
\dd{\mu_{\be, g, h}}{\mu_h}
= \fc{\E_{\mu_h} \exp(H^{\ge 2}(x))}{\exp(H^{\ge 2}(x))} \le \exp \pa{2n C_2 \sqrt{C_0(\be)}}
\]
and $\KL(\mu_{\be, g, h}\|\mu_h) = \E_{\mu_{\be, g, h}} \ln \pa{\dd{\mu_{\be, g, h}}{\mu_h}}\le 2n C_2\sqrt{C_0(\be)}$. 
\end{proof}

\begin{proof}[Proof of~\cref{t:main-p-spin}]
    By Theorem~\ref{t:atoe-p-spin} (approximate tensorization for $p$-spin model) and Theorem~\ref{t:k-glauber} (speedup of $k$-Glauber), there is $C_0$ depending on $D(\be)$ such that with probability $\ge 1-e^{-\Om(n)}$ over $g$,  
for any $s$, $t\cdot \fc ns$ steps of $s$-Glauber dynamics 
results in a distribution $\nu_t$ satisfying
\begin{align*}
    \TV(\nu_t\|\mu_{\be, g, h}) \le 
    \sqrt{\rc 2 \KL(\nu_t\|\mu_{\be, g, h})} 
    \le \sqrt{\rc 2 \KL(\nu_0\|\mu_{\be, g, h}) e^{-C_0t}}.
\end{align*}
(Alternatively, note that the same local-to-global argument that shows approximate tensorization for the $p$-spin model can be used to show contraction for $\Dn{k}{k-\ell}$ instead of $\Dn{k}{k-1}$,  \cite[Theorem 38]{anari2023universality}, \cite[Theorem 20]{anari2022optimal}.)
With the product initialization $\nu_0=\mu_h$, we have Lemma~\ref{l:p-spin-init} that with probability $1-e^{-\Om(n)}$, $\KL(\nu_0\|\mu_{J,h})\le O(C(\be)n)$. Hence there exists $C_2=C_2'A/C_0$ such that under these good events, with $T = \fc{C_2}{c_1}\ln \pf n\ep \sqrt n$ steps of $c_1\sqrt n$-Glauber dynamics, $\TV(\nu_T\|\mu_{\be, g, h})\le 
\fc{\ep}{2}$. 

By a coupling argument, the actual run of Algorithm~\ref{a:ppspin} incurs additional error equal to the sum of the errors from each call of $\qars$. By choosing $\estep\le \fc{\ep}{2T}$, the total TV error is at most $\ep$. 
Finally, by trying $O\prc{\estep}$ proposals in parallel in $\mathsf{ApproxRejectionSampler}$, we can ensure that with probability $\ge 1-\ep$, the total time is at most $O\pa{\ln \prc{\ep}}$ times the number of calls to $\mathsf{ApproxRejectionSampler}$. 
\end{proof}



%% file: comparison.tex
\section{Comparison with \cite{liu2022simple}}
\label{s:comp}
We show that the result of \cite{liu2022simple} also gives \Cref{t:main} with a different algorithm, but cannot be used to derive \Cref{t:main-p-spin}. 
Rather than define a different Markov chain as in our work, \cite{liu2022simple} consider the original (continuous-time) Glauber dynamics and give a faster way to simulate a block of updates in parallel by a belief propagation type algorithm. Their algorithm is a Las Vegas algorithm in that it is guaranteed to faithfully simulate Glauber dynamics, with a high-probability bound on the parallel running time, so the desired accuracy $\ep$ does not need to be given in advance. Our algorithm does need $\ep$ as the proposal distribution introduces bias which needs to be controlled.
\begin{df}
Let $\mu$ be a distribution on $\prodo in \Om_i$.
    Define the \vocab{Dobrushin influence matrix} of $\mu$ (for Glauber dynamics) by $\rh\in \R^{n\times n}$ where
    \[
\rh_{ij} := 
\max_{\tau, \tau' \in \prod_{k\ne j} \Om_k: \tau_{\sim i} = \tau'_{\sim i}} \TV( \mu(X_j = \cdot | X_{\sim j }= \tau) , \mu(X_j = \cdot | X_{\sim j }= \tau') ).
    \]
\end{df}
For Ising models, we can bound the Dobrushin matrix as follows.
\begin{pr}
\label{e:dobrushin}
Suppose $J\in \R^{n\times n}$ is symmetric. 
    The Dobrushin influence matrix of $\mu_{J,h}$ has entries $\rh_{ij}\le |J_{ij}|$. Therefore, $\ve{\rh}_2 \le \ve{J}_F$. 
\end{pr}
\begin{proof}
    We have that for $\tau_{\sim i} = \tau'_{\sim i}$, $\tau_i=1$, $\tau'_i=-1$ that 
    \begin{align*}
& \TV( \mu(X_j = \cdot | X_{\sim j }= \tau) , \mu(X_j = \cdot | X_{\sim j }= \tau') )\\
& = \rc 2 \ab{
\pa{(\mu(X_j = 1 | X_{\sim j }= \tau) - \mu(X_j = -1 | X_{\sim j }= \tau))
 - (\mu(X_i = 1 | X_{\sim j }= \tau') - \mu(X_j = -1 | X_{\sim j }= \tau'))
 }} \\
&= \rc 2 \ab{\tanh \pa{\sum_{k\ne j} J_{jk} \tau_k + h_j} - \tanh \pa{\sum_{k\ne j} J_{jk} \tau_k' + h_j}} \\
&\le \rc 2 |2 J_{ji}| = |J_{ij}|
    \end{align*}
where we use the fact that $\tanh$ is 1-Lipschitz and 
$\sum_{k\ne j} J_{jk} \tau_k - \sum_{k\ne j} J_{jk} \tau_k' = 2J_{ji}$. Taking the maximum over $\tau, \tau'$ gives the bound on $\rh_{ij}$. 
The second result follows from $\ve{\rh}_2 \le \ve{\rh}_F \le \ve{J}_F$. 
\end{proof}

Define the continuous-time Glauber dynamics by associating each site $v\in V$ with a rate-1 Poisson clock and update the spin of $v$ according to Glauber dynamics whenever it rings.
\begin{thm}[{\cite[Theorem 1.2, 5.1]{liu2022simple}}]
\label{t:ly22}
    Consider a positive undirected graphical model $\mu$ on $\Om^V$ with graph $G=(V,E)$, $|V|=n$, $|E|=m$. 
    Assume that for some $p\in [1,\iy]$ that $\ve{\rh}_p\le C$.
    Consider continuous-time Glauber dynamics on $V$. 
    There is a Las Vegas parallel algorithm that given $T=O(n^{O(1)})$, outputs $X_T$ with  $O(CT+\ln n)$ depth (time) using $\wt O(m + n|\Om|^2)$ processors.
\end{thm}
Note that the main theorem in \cite{liu2022simple} is stated for $C=O(1)$, but an examination of the proof shows that it works for general $C$ with the extra $C$ factor in the depth. This follows from keeping the $C$ factor in their Lemma 3.4 in the proof of Theorem 5.1.

Combining~\Cref{e:dobrushin} and~\Cref{t:ly22} with the fast $O\pa{\fc{1}{1-\ve{J}}\ln \pf{n}{\ep}}$ mixing time of continuous-time Glauber dynamics for the Ising model with $\ve{J}<1$ \cite{anari2021entropic} gives~\Cref{t:main}.\footnote{\cite{anari2021entropic} consider the discrete-time chain; the mixing times for the discrete-time and continuous-time chains differ by a factor of $n$ \cite[\S2]{liu2022simple}.} 

We note however, that the bound on the Dobrushin influence matrix fails to be sublinear in $n$ for the $p$-spin model, so that their results are unable to derive \Cref{t:main-p-spin}. Consider a 3-spin model
\[
\mu_{\be, g,h}(x) \propto 
\exp(H_{\be, g, h}(x)), 
\quad \text{where }
H_{\be, g, h}(x) = \fc{\be\sqrt{6}}{n} \sum_{1\le i_1<i_2<i_3\le n} g_{\{i_1,i_2,i_3\}}x_{i_1}x_{i_2}x_{i_3} + \sumo in h_ix_i 
\]
where $g_{\{i_1,i_2,i_3\}}\sim N(0,1)$. 
As before, we can bound (switching indices for convenience), for $\tau_{\sim j} = \tau'_{\sim j}$, $\tau_j=1$, $\tau'_j=-1$,
\begin{align*}
    &\TV( \mu(X_i = \cdot | X_{\sim i }= \tau) , \mu(X_i = \cdot | X_{\sim i }= \tau') )\\
    &= \rc 2 \ab{
    \tanh \pa{\fc{\be \sqrt 6}n\sum_{k_1<k_2,\,i\nin \{k_1,k_2\}} g_{\{i,k_1,k_2\}} \tau_{k_1} \tau_{k_2} + h_i }
    - \tanh \pa{\fc{\be \sqrt 6}n\sum_{k_1<k_2,\,i\nin \{k_1,k_2\}} g_{\{i,k_1,k_2\}} \tau'_{k_1} \tau'_{k_2} + h_i }
    } \\
    &\le \rc 2 \cdot \fc{\be \sqrt 6}{n} 
    \ab{2\sum_{k\ne i,j} g_{\{i,j,k\}}\tau_k}
\end{align*}
However, choosing $\tau_k = \sgn (g_{\{i,j,k\}})$ and noting that $g_{\{i,j,k\}}$ is order 1, we obtain a bound of order 1 for each entry of $\rh$. Considering $h=0$, we can also consider the linearization of $\tanh$ around 0 to show that in fact $\rh_{ji} = \Te\pa{\max_{\tau\in \{\pm 1\}^{[n]\bs \{i,j\}}} \fc{\be}n \ab{\sum_{k\ne i,j} g_{\{i,j,k\}}\tau_k}\wedge 1} = \Te(1)$ with high probability. Under this event, $\ve{\rh}_p$ is of order 1 for any $p\in [1,\iy]$. 